\documentclass[]{interact}

\usepackage{epstopdf}
\usepackage[caption=false]{subfig}

\usepackage[natbibapa,nodoi]{apacite}
\usepackage{algorithm}
\usepackage{algorithmic}
\usepackage{graphicx, caption}
\usepackage{tikz}
\usetikzlibrary{shapes,arrows}
\usetikzlibrary{arrows.meta,
                positioning}

\newtheorem{assumption}{Assumption}
\theoremstyle{plain}
\newtheorem{theorem}{Theorem}[section]
\newtheorem{lemma}[theorem]{Lemma}

\theoremstyle{definition}

\theoremstyle{remark}
\newtheorem{remark}{Remark}

\begin{document}

\articletype{ARTICLE TEMPLATE}

\title{Fast data-driven iterative learning control for linear system with output disturbance}

\author{
\name{Jia Wang \textsuperscript{a}\thanks{CONTACT Jia Wang. Email: jia.wang@esat.kuleuven.be}, Leander Hemelhof \textsuperscript{a}, Ivan Markovsky \textsuperscript{b,c} and Panagiotis Patrinos \textsuperscript{a}}
\affil{\textsuperscript{a} Department of Electrical Engineering (ESAT-STADIUS), KU Leuven, Kasteelpark Arenberg 10, 3001 Leuven, Belgium
\\
\textsuperscript{b} Credible Data-driven Models Group, International Centre for Numerical Methods in Engineering (CIMNE), Gran Capità, 08034 Barcelona, Spain
\\
\textsuperscript{c} Catalan Institution for Research and Advanced Studies (ICREA), Pg. Lluis Companys 23, 08010 Barcelona, Spain
}
}

\maketitle

\begin{abstract}
This paper studies data-driven iterative learning control (ILC) for linear time-invariant (LTI) systems with unknown dynamics, output disturbances and input box-constraints. Our main contributions are: 1) using a non-parametric data-driven representation of the system dynamics, for dealing with the unknown system dynamics in the context of ILC, 2) design of a fast ILC method for dealing with output disturbances, model uncertainty and input constraints. A complete design method is given in this paper, which consists of the data-driven representation, controller formulation, acceleration strategy and convergence analysis. A batch of numerical experiments and a case study on a high-precision robotic motion system are given in the end to show the effectiveness of the proposed method.
\end{abstract}


\section{Introduction}

Iterative learning control (ILC) was proposed for repetitive control tasks and has been widely used to improve the tracking accuracy in many industries, for example printing systems \citep{Rozario_2019}, atomic force microscopes \citep{Tien_2005} and continuously stirred tank reactors \citep{Chi_2018}. By using the previous tracking errors and control signals, ILC updates the current control signal to steer the system output to be as close as possible to the reference trajectory over a ﬁxed time interval.

The existing solutions of ILC mainly focus on Lyapunov function-based adaptive ILC \citep{Tayebi_2004}, contraction mapping-based proportional derivative ILC \citep{Park_1999} and optimality-based ILC \citep{Barton_2011}. Optimality-based ILC is the most popular choice for real-world applications since it can reject the undesirable large transient behavior \citep{Chi_2018}. However, it has to use the model information in the algorithm design, e.g., the impulse response or state-space representation, hence, it is classified as a model-based control method \citep{Chi_2015}. In the case that the nominal model is not given or ﬁrst-principle modelling is expensive to use, system identification is required. One of the most powerful methods is subspace identification \citep{Overschee_2012}, which has been used successfully in the ILC scenario \citep{Nijsse_2001,Janczak_2013}, but only the white Gaussian disturbance is considered in the existing literature. In practice, the system output can be perturbed with structured disturbances such as the sine wave \citep{Markovsky_2020} and periodic disturbances \citep{Bodson_2001}. Therefore, the data-driven representation method that can handle the structured disturbance effectively is worth investigating. Recently, Willems et al.’s fundamental lemma \citep{Willems_2005} has received a great attention in the control community, which offers a non-parametric data-driven representation of the system dynamics based on available input-output measurements \citep{Markovsky_2008,Markovsky_2021}, and has been successfully used in many applications \citep{Elokda_2021,Carlet_2020}. Authors in the recent work \citep{Jiang_2022} propose a novel data-driven ILC for linear time-invariant (LTI) systems with the aid of Willems et al.’s fundamental lemma, which is then extended to point-to-point ILC tasks \citep{Jiang_2023}. In \cite{Wang_2023}, a modiﬁed Willems et al.’s fundamental lemma is proposed to relax assumptions on the controllability and persistency of excitation, and then combined with the optimality-based ILC for LTI systems. However, the methods above assume no output disturbances. Inspired by \cite{Jiang_2022}, we apply a non-parametric data-driven representation to the optimality-based ILC under output disturbances, which can be promising to handle structured disturbances in the repetitive control process, due to its non-parametric nature. This constitutes the first motivation of this work.

On the other hand, the convergence speed is one of the most important performance indicators for the ILC. In recent years, there are many attempts to design ILC algorithms with a faster convergence speed. In \cite{Harte_2005,Ratcliffe_2008}, the classical gradient-based ILC is designed and robust monotonic convergence under multiplicative uncertainty is established. In \cite{Chu_2009}, an accelerated version is proposed, which uses successive projection and a momentum term in the control signal update. In \cite{Geng_2015}, a Quasi-Newton ILC method is proposed to improve the convergence speed and convergence is proved. Moreover, the accelerated gradient method \citep{Nesterov_1983} is combined with the gradient-based ILC in \cite{Gu_2019}, and a faster convergence speed is achieved. However, the above mentioned accelerated methods assume to know the exact model and no output disturbances, which is impractical in many applications. Therefore, it is important to develop a fast ILC method based on the available, but inexact information. In \cite{Devolder_2014}, an inexact accelerated gradient method is proposed, which is an extension of the exact method in \cite{Nesterov_2005}. Inspired by \cite{Devolder_2014}, we attempt to combine the inexact accelerated method with the classical optimality-based ILC, and present the convergence rate under an inexact model and output disturbances, which is the second motivation of this paper.

Motivated by the above discussions, non-parametric data-driven representation based fast ILC is investigated in this paper, which aims to improve the ILC convergence speed and robustness under unknown and inexact information. The main contributions are summed up as follows:

\begin{itemize}
  \item Non-parametric data-driven ILC is developed, which can deal efficiently with structured disturbances.
  \item The fast ILC approach with inexact information is developed. Moreover, the convergence rate of applying the inexact Nesterov accelerated gradient method into the ILC scenario is provided.
  \item A hybrid ILC approach is proposed with an empirical but effective switching condition between the fast and classical ILC, which inherits the advantages of these two methods, i.e., faster speed and better asymptotic performance.
\end{itemize}

The remainder of this paper is organized as follows. Section \ref{sec2} gives the problem statement. Section \ref{sec3} formulates the data-driven representation. In Section \ref{sec4}, fast ILC methods with inexact information are developed. Numerical experiments and a case study are given to verify the effectiveness of the proposed ILC framework in Section \ref{sec5}. Finally, Section \ref{sec6} concludes this paper.

\section{Problem statement}\label{sec2}

Consider the following controllable single-input single-output (SISO) LTI system:

\begin{small}
\begin{equation}\label{LTI}
\begin{cases}
  x_{j}(k+1) = Ax_{j}(k)+Bu_{j}(k)\\
  y_{j}(k) = Cx_{j}(k)
\end{cases}
\end{equation}
\end{small}where $j$ denotes the trial index, and $k\in\{0\ldots N\}$ denotes the time index. $N$ is the data length in one trial. $x_{j}(k)\in\mathbb{R}^{n}$, $u_{j}(k)\in Q:=[\underline{u},\overline{u}]\subset\mathbb{R}$ and $y_{j}(k)\in\mathbb{R}$ represent the system state, input and output, respectively. $\underline{u}$ and $\overline{u}$ are the lower and upper bound on $u_{j}(k)$. $A$, $B$, $C$ are real but unknown matrices with appropriate dimensions. The initial state is the same for all trials, i.e., $x_{j}(0)=x_{0}$ and assumed to be unknown.

Taking into account (\ref{LTI}) with $k=0\ldots N$, the relation between input and output is

\begin{small}
\begin{equation*}
\begin{split}
y_{j}(1)&=Cx_{j}(1)\\
&=CAx_{j}(0)+CBu_{j}(0) \\
y_{j}(2)&=Cx_{j}(2)=CAx_{j}(1)+CBu_{j}(1)\\
&=CA^{2}x_{j}(0)+CABu_{j}(0)+CBu_{j}(1) \\
&\ \vdots \\
y_{j}(N)&=Cx_{j}(N)=CAx_{j}(N-1)+CBu_{j}(N-1)\\
&=CA^{N}x_{j}(0)+CA^{N-1}Bu_{j}(0)+CA^{N-2}Bu_{j}(1)\\
&\ \ \ +\cdots+CBu_{j}(N-1)
\end{split}
\end{equation*}
\end{small}Then the input/output mapping in the trial-domain can be written as

\begin{small}
\begin{equation}\label{lifted_system}
y_{j}
=
Gu_{j}+c
\end{equation}
\end{small}where

\begin{small}
\begin{subequations}\label{shifted_trajectory}
\begin{align}
& u_{j}
=
\begin{bmatrix}
  u_{j}(0) & u_{j}(1) & \ldots & u_{j}(N-1)
\end{bmatrix}^{T} \\
& y_{j}
=
\begin{bmatrix}
  y_{j}(1) & y_{j}(2) & \ldots & y_{j}(N)
\end{bmatrix}^{T}
\end{align}
\end{subequations}
\end{small}and

\begin{small}
\begin{equation}\label{convolutional_matrix}
G
=
\begin{bmatrix}
  CB & 0 & \ldots & 0 \\
  CAB & CB & \ddots & \vdots \\
  \vdots & \ddots & \ddots & 0 \\
  CA^{N-1}B & \ldots & CAB & CB
\end{bmatrix},\ \
c
=
\begin{bmatrix}
  CA \\
  CA^{2} \\
  \vdots \\
  CA^{N}
\end{bmatrix}
x_{0}
\end{equation}
\end{small}with $u_{j}\in\mathcal{Q}$ and $\mathcal{Q}=Q\times Q\times\cdots\times Q$ is the Cartesian product. The following assumption is given to guarantee the invertibility of $G$.

\begin{assumption}\label{assumption_invertible}
The system (\ref{LTI}) has relative degree $1$, i.e., $CB\neq0$.
\end{assumption}

For a given reference trajectory $r=\begin{bmatrix}r(1) & \ldots &  r(N)\end{bmatrix}^{T}$, the control objective is to find the desired input $u^{*}$ that is a solution of the following constrained least-squares problem:

\begin{small}
\begin{equation}\label{original_LS}
\min\limits_{u\in\mathcal{Q}}J(u)
=
\frac{1}{2}\|r-(Gu+c)\|_{2}^{2}
\end{equation}
\end{small}in which the gradient of the cost function $J(u)$ is

\begin{small}
\begin{equation}\label{original_gradient}
\nabla J(u)
=
-G^{T}
\begin{pmatrix}
r-(Gu+c)
\end{pmatrix}
\end{equation}
\end{small}Since the box-constraint set $\mathcal{Q}$ has a simple structure, an effective way to solve (\ref{original_LS}) is the projected gradient method:

\begin{small}
\begin{equation}\label{original_projection_gradient}
u_{j+1}
=
\Pi_{\mathcal{Q}}
\begin{pmatrix}
u_{j}
-
\rho
\nabla J(u_{j})
\end{pmatrix}
\end{equation}
\end{small}where $\rho\in(0,2/L)$ is the step-size, $L=\lambda_{\max}(G^{T}G)$ is the Lipschitz constant of $J(u)$, $\lambda_{\max}(\cdot)$ denotes the maximal eigenvalue and $\Pi_{\mathcal{Q}}(\cdot)$ is the Euclidean projection on the set $\mathcal{Q}$. If the input is unconstrained, we have $u^{*}=(G^{T}G)^{-1}G^{T}(r-c)=G^{-1}(r-c)$ since $G$ is invertible for a SISO system according to Assumption \ref{assumption_invertible}. Note that the iteration number $j$ in (\ref{original_projection_gradient}) is also the trial index in the ILC sense.

In practice, the exact representation matrix $G$ is usually unknown and thus some system identification techniques are required. However, since the output channel is perturbed by the measurement disturbance, only $\tilde{y}=y+d$ can be collected as offline data, in which $d=\begin{bmatrix}d(1) & \ldots & d(T)\end{bmatrix}^{T}$ and $T$ is the offline data length, then the matrix $\tilde{G}$ obtained from the system identification is inexact. Moreover, only $\tilde{y}_{j}=y_{j}+d_{j},j\geq0$ can be measured in the online control, with $d_{j}=\begin{bmatrix}d_{j}(1) & \ldots & d_{j}(N)\end{bmatrix}^{T}$, which leads to an inexact gradient as $\tilde{\nabla}J(u_{j})=-\tilde{G}^{T}(r-\tilde{y}_{j})$ at the $j$-th trial. Even if we have the exact representation matrix $G$, no disturbances and no constraints on $u$, $G$ can be ill-conditioned and thus sensitive to system inversion \citep{Gu_2019}. On the other hand, fast convergence speed of the learning controller is strongly demanded for many real-time applications \citep{Geng_2015}. The reasons above show that formula (\ref{original_projection_gradient}) may be limited in our setting. Therefore, the goal of this paper is to design a fast learning algorithm to generate an input sequence $\{u_{j}\}$ such that

\begin{small}
\begin{equation*}
\lim\limits_{j\rightarrow\infty}\|\varepsilon_{j}\|_{2}\leq\tau_{\varepsilon},\
\lim\limits_{j\rightarrow\infty}\|u^{*}-u_{j}\|_{2}\leq\tau_{u}
\end{equation*}
\end{small}with the inexact information $e_{j}$ and $\tilde{\nabla}J(u_{j})$, in which $\varepsilon_{j}=r-y_{j}$ and $e_{j}=r-\tilde{y}_{j}$ is the exact and inexact tracking error, respectively. Positive numbers $\tau_{\varepsilon}$ and $\tau_{u}$ are given accuracy levels.

\section{Data-driven representation}\label{sec3}

In this paper, a non-parametric data-driven representation of the system dynamics is used. Some basic formulations are given in the following subsection, see \cite{Markovsky_2021} for details. The noise filtering technique is introduced in the second subsection in order to obtain a more precise representation $\tilde{G}$.

\subsection{Basic formulation}

Let $w_{\mathrm{d}}=\begin{bmatrix}w_{\mathrm{d}}(0) & \ldots & w_{\mathrm{d}}(T)\end{bmatrix}^{T}$ be a $(T+1)$-sample long offline data trajectory of the SISO LTI system (\ref{LTI}), in which $w_{\mathrm{d}}(i)=\begin{bmatrix}u_{\mathrm{d}}(i) & y_{\mathrm{d}}(i)\end{bmatrix}^{T}$ and the subscript `$\mathrm{d}$' denotes `data'. Then the $2(K+1)\times(T-K+1)$-dimensional Hankel matrix for the data trajectory is formed by

\begin{small}
\begin{equation}\label{Hankel}
H_{K+1}(w_{\mathrm{d}})
=
\begin{bmatrix}
  w_{\mathrm{d}}(0) & w_{\mathrm{d}}(1) & \ldots & w_{\mathrm{d}}(T-K) \\
  w_{\mathrm{d}}(1) & w_{\mathrm{d}}(2) & \ldots & w_{\mathrm{d}}(T-K+1) \\
  \vdots & \vdots & \ddots & \vdots \\
  w_{\mathrm{d}}(K) & w_{\mathrm{d}}(K+1) & \ldots & w_{\mathrm{d}}(T)
\end{bmatrix}
\end{equation}
\end{small}Permute the Hankel matrix $H_{K+1}(w_{\mathrm{d}})$ and partition as

\begin{small}
\begin{equation}\label{partition_U_Y}
H_{K+1}(u_{\mathrm{d}})
=
\begin{bmatrix}
  U_{\mathrm{p}} \\
  U_{\mathrm{f}} \\
  U_{\mathrm{b}}
\end{bmatrix}
,
\ \ \
H_{K+1}(y_{\mathrm{d}})
=
\begin{bmatrix}
  Y_{\mathrm{b}} \\
  Y_{\mathrm{p}} \\
  Y_{\mathrm{f}}
\end{bmatrix}
\end{equation}
\end{small}where $U_{\mathrm{p}},Y_{\mathrm{p}}\in\mathbb{R}^{T_{\mathrm{ini}}\times(T-K+1)}$, $U_{\mathrm{f}},Y_{\mathrm{f}}\in\mathbb{R}^{N\times(T-K+1)}$ and $U_{\mathrm{b}},Y_{\mathrm{b}}\in\mathbb{R}^{1\times(T-K+1)}$.  $K=T_{\mathrm{ini}}+N$ and $T_{\mathrm{ini}}\geq l$ is the initial trajectory length, in which $l$ is the lag of system (\ref{LTI}). Then define $\mathcal{H}_{K}(w_{\mathrm{d}}):=\begin{bmatrix}U_{\mathrm{p}}^{T} & Y_{\mathrm{p}}^{T} & U_{\mathrm{f}}^{T} & Y_{\mathrm{f}}^{T}\end{bmatrix}^{T}\in\mathbb{R}^{2K\times(T-K+1)}$. Note that the lifted model (\ref{convolutional_matrix}) yields a time shift in trajectories $u_{j}$ and $y_{j}$, see (\ref{shifted_trajectory}), therefore, $u_{\mathrm{d}}(T)$ and $y_{\mathrm{d}}(0)$ should be removed from the data trajectory, or equivalently, remove the blocks $U_{\mathrm{b}}$ and $Y_{\mathrm{b}}$ in (\ref{partition_U_Y}) when constructing the matrix $\mathcal{H}_{K}(w_{\mathrm{d}})$.

Choose the input trajectory $\begin{bmatrix}u_{\mathrm{d}}(0) & \ldots & u_{\mathrm{d}}(T-1)\end{bmatrix}^{T}$ that is persistently exciting of order $K+n$, then according to the fundamental lemma \citep{Willems_2005}, there exists a vector $g\in\mathbb{R}^{T-K+1}$ such that

\begin{small}
\begin{equation}\label{linear_combination}
\mathcal{H}_{K}(w_{\mathrm{d}})
g
=
\begin{bmatrix}
  u_{\mathrm{ini}} \\
  y_{\mathrm{ini}} \\
  u_{\mathrm{f}} \\
  y_{\mathrm{f}}
\end{bmatrix}
\end{equation}
\end{small}The trajectory in the right hand side of (\ref{linear_combination}) follows the lifted system (\ref{lifted_system}), i.e., $y_{\mathrm{f}}=Gu_{\mathrm{f}}+c(u_{\mathrm{ini}},y_{\mathrm{ini}})$, in which $u_{\mathrm{ini}},y_{\mathrm{ini}}\in\mathbb{R}^{T_{\mathrm{ini}}}$ is the initial condition, $c(u_{\mathrm{ini}},y_{\mathrm{ini}})$ is the initial response related with $u_{\mathrm{ini}},y_{\mathrm{ini}}$, and trajectories $u_{\mathrm{f}},y_{\mathrm{f}}\in\mathbb{R}^{N}$ is to be determined.

Since the system is assumed to be controllable, the initial input trajectory $u_{\mathrm{ini}}$ determines $x_{0}$ as follows:

\begin{small}
\begin{equation}
x_{0}
=
\begin{bmatrix}
  A^{T_{\mathrm{ini}}-1}B\ \ \ &A^{T_{\mathrm{ini}}-2}B & \ldots & B
\end{bmatrix}
\begin{bmatrix}
  u_{\mathrm{ini}}(0) \\
  u_{\mathrm{ini}}(1) \\
  \vdots \\
  u_{\mathrm{ini}}(T_{\mathrm{ini}}-1)
\end{bmatrix}
\end{equation}
\end{small}which steers the system from the origin to $x_{0}$, and $y_{\mathrm{ini}}=\begin{bmatrix}y_{\mathrm{ini}}(1) & \ldots & y_{\mathrm{ini}}(T_{\mathrm{ini}})\end{bmatrix}^{T}$ is the corresponding output trajectory. In this paper, the initial condition $(u_{\mathrm{ini}},y_{\mathrm{ini}})$ is assumed to be unknown.

According to the generalized persistency of excitation condition in \cite{Markovsky_2023}, for $K\geq l$, the range of $\mathcal{H}_{K}(w_{\mathrm{d}})$ captures all possible trajectories of (\ref{lifted_system}) if and only if $rank\begin{pmatrix}\mathcal{H}_{K}(w_{\mathrm{d}})\end{pmatrix}=K+n$. However, if the output is perturbed by disturbances, the Hankel matrix $\mathcal{H}_{K}(w_{\mathrm{d}})$ will have full row rank, which violates the above rank condition, such that the system behavior cannot be represented \citep{Markovsky_2023}. In order to deal with the noise, we apply a noise filtering technique in the next subsection.

\subsection{Noise filtering}

In this part, a heuristic method based on the singular value decomposition (SVD) and low rank approximation is used to filter the perturbed Hankel matrix. Specifically, define the matrix with output disturbances as

\begin{small}
\begin{equation}\label{original_Hankel}
\begin{bmatrix}
  U \\
  \tilde{Y}
\end{bmatrix}
:=
\begin{bmatrix}
  U_{\mathrm{p}} \\
  U_{\mathrm{f}} \\
  \tilde{Y}_{\mathrm{p}} \\
  \tilde{Y}_{\mathrm{f}}
\end{bmatrix}
\end{equation}
\end{small}The construction of $\tilde{Y}_{\mathrm{p}}$ and $\tilde{Y}_{\mathrm{f}}$ are similar to (\ref{partition_U_Y}) but with a perturbed trajectory $\tilde{y}_{\mathrm{d}}$. In order to denoise the perturbed part $\tilde{Y}$, we follow the approach of \cite{Demmel_1987}. Specifically, first apply the QR decomposition on the input part as

\begin{small}
\begin{equation}
U^{T}
=
\begin{bmatrix}
  Q_{1} & Q_{2}
\end{bmatrix}
\begin{bmatrix}
  R_{1} \\
  R_{2}
\end{bmatrix}
\end{equation}
\end{small}where $\begin{bmatrix}Q_{1} & Q_{2}\end{bmatrix}$ is an orthogonal matrix, $R_{1}\in\mathbb{R}^{K\times K}$ is an upper triangular matrix and $R_{2}\in\mathbb{R}^{K\times(T-2K+1)}$ is the zero matrix $\mathbf{0}$. Left multiplying $\begin{bmatrix}Q_{1} & Q_{2}\end{bmatrix}^{T}$ with $\begin{bmatrix}U^{T} & \tilde{Y}^{T}\end{bmatrix}$ results in

\begin{small}
\begin{equation}\label{QD_UY}
\begin{bmatrix}
  Q_{1}^{T} \\
  Q_{2}^{T}
\end{bmatrix}
\begin{bmatrix}
  U^{T} & \tilde{Y}^{T}
\end{bmatrix}
=
\begin{bmatrix}
  R_{1} & D_{1} \\
  \mathbf{0} & D_{2}
\end{bmatrix}
\end{equation}
\end{small}In this way, the first block row in (\ref{QD_UY}) has full row rank. In order to reduce the rank of (\ref{original_Hankel}) from $2K$ to $K+n$, apply SVD on the matrix $D_{2}$ as $D_{2}=\mathcal{U}\Sigma\mathcal{V}^{T}$, then the rank $n$ approximation is obtained by

\begin{small}
\begin{equation}\label{svd_n}
\bar{D}_{2}
=
\sum_{i=1}^{n}
\sigma_{i}
\mathcal{U}_{i}
\mathcal{V}_{i}^{T}
\end{equation}
\end{small}where $\sigma_{i}$ is the $i$-th singular value stored in the matrix $\Sigma$. $\mathcal{U}_{i}$ and $\mathcal{V}_{i}$ are the $i$-th column in the matrix $\mathcal{U}$ and $\mathcal{V}$. Then construct the output matrix as

\begin{small}
\begin{equation}
\bar{Y}
=
\begin{pmatrix}
\begin{bmatrix}
  Q_{1} & Q_{2}
\end{bmatrix}
\begin{bmatrix}
  D_{1} \\
  \bar{D}_{2}
\end{bmatrix}
\end{pmatrix}^{T}
\end{equation}
\end{small}and we have $rank\begin{pmatrix}\begin{bmatrix}U^{T} & \bar{Y}^{T}\end{bmatrix}\end{pmatrix}=K+n$. Partition this low rank matrix and define $\mathcal{H}_{K}(\bar{w}_{\mathrm{d}}):=\begin{bmatrix}U_{\mathrm{p}}^{T} & \bar{Y}_{\mathrm{p}}^{T} & U_{\mathrm{f}}^{T} & \bar{Y}_{\mathrm{f}}^{T}\end{bmatrix}^{T}$, then the method introduced in the above subsection can be used. To be specific, there exists a vector $\bar{g}\in\mathbb{R}^{T-K+1}$ such that

\begin{small}
\begin{equation}\label{g_evolution}
\mathcal{H}_{K}(\bar{w}_{\mathrm{d}})
\bar{g}
=
\begin{bmatrix}
  u_{\mathrm{ini}} \\
  y_{\mathrm{ini}} \\
  u_{j} \\
  y_{j}
\end{bmatrix}
\end{equation}
\end{small}for the lifted system (\ref{lifted_system}), in which the subscript `$j$' denotes the trial index. In this way, $\bar{g}$ can be written as

\begin{small}
\begin{equation}
\bar{g}
=
\begin{bmatrix}
  U_{\mathrm{p}} \\
  \bar{Y}_{\mathrm{p}} \\
  U_{\mathrm{f}}
\end{bmatrix}^{\dag}
\begin{bmatrix}
  u_{\mathrm{ini}} \\
  y_{\mathrm{ini}} \\
  u_{j}
\end{bmatrix}
=
\Phi
u_{j}
+
\Psi
\begin{bmatrix}
  u_{\mathrm{ini}} \\
  y_{\mathrm{ini}}
\end{bmatrix}
\end{equation}
\end{small}where $\dag$ denotes the pseudo-inverse. $\Phi\in\mathbb{R}^{(T-K+1)\times N}$ and $\Psi\in\mathbb{R}^{(T-K+1)\times 2T_{\mathrm{ini}}}$ are the partitioning of $\begin{pmatrix}\begin{bmatrix}U_{\mathrm{p}}^{T} & \bar{Y}_{\mathrm{p}}^{T} & U_{\mathrm{f}}^{T}\end{bmatrix}^{T}\end{pmatrix}^{\dag}$ corresponding to $u_{j}$ and $\begin{bmatrix}u_{\mathrm{ini}}^{T} & y_{\mathrm{ini}}^{T}\end{bmatrix}^{T}$. Then, the output trajectory is

\begin{small}
\begin{equation}\label{output_trajectories}
y_{j}
=
\bar{Y}_{\mathrm{f}}
\bar{g}
=
\bar{Y}_{\mathrm{f}}
\Phi
u_{j}
+
\bar{Y}_{\mathrm{f}}
\Psi
\begin{bmatrix}
  u_{\mathrm{ini}} \\
  y_{\mathrm{ini}}
\end{bmatrix}
\end{equation}
\end{small}The above formulation has a similar structure to (\ref{lifted_system}), i.e., $\bar{Y}_{\mathrm{f}}\Phi$ corresponds to the representation matrix $G$ and $\bar{Y}_{\mathrm{f}}\Psi\begin{bmatrix}u_{\mathrm{ini}}^{T} & y_{\mathrm{ini}}^{T}\end{bmatrix}^{T}$ corresponds to the initial response. On the other hand, the data-driven representation matrix should have the same lower triangular structure as $G$ due to causality, therefore, a projection operator $\mathcal{P}_{0}(\cdot)$ is applied on $\bar{Y}_{\mathrm{f}}\Phi$ in order to impose elements above the diagonal as $0$. In this way, we obtain the data-driven representation from the noisy data:

\begin{small}
\begin{equation}\label{Willems_fundamental_lemma_model}
\tilde{G}=\mathcal{P}_{0}(\bar{Y}_{\mathrm{f}}\Phi)
\end{equation}
\end{small}Then the inexact representation (\ref{Willems_fundamental_lemma_model}) can be used to evaluated the function value and gradient. In the next section, the above inexact information is applied in the classical and accelerated gradient methods.

\section{ILC Design}\label{sec4}

In the case of no offline and online disturbances, the function value and gradient of (\ref{original_LS}) can be evaluated exactly:

\begin{small}
\begin{equation}\label{exact_information}
u_{j}
\rightarrow
\begin{cases}
  J(u_{j})=\frac{1}{2}\|\varepsilon_{j}\|_{2}^{2}\\
  \nabla J(u_{j})=-G^{T}\varepsilon_{j}
\end{cases}
\end{equation}
\end{small}where $\varepsilon_{j}$ is the tracking error without disturbances, i.e., $\varepsilon_{j}=r-y_{j}=r-(Gu_{j}+c)$. However, in practice, only inexact information is available:

\begin{small}
\begin{equation}\label{inexact_information}
u_{j}
\rightarrow
\begin{cases}
  \tilde{J}(u_{j})=\frac{1}{2}\|e_{j}\|_{2}^{2}\\
  \tilde{\nabla} J(u_{j})=-\tilde{G}^{T}e_{j}
\end{cases}
\end{equation}
\end{small}where $e_{j}$ is defined as $e_{j}=r-\tilde{y}_{j}=r-(Gu_{j}+c+d_{j})$. Therefore, the ILC algorithm with the classical (or accelerated) gradient and inexact information (\ref{inexact_information}) is provided to approximate $u^{*}$, which is the optimal solution of (\ref{original_LS}).

The following assumption is first introduced to bound the disturbance:

\begin{assumption}\label{assumption_disturbance}
Each element of the unknown disturbance trajectory is bounded, i.e., $\exists\ \bar{d}>0$ such that $|d_{j}(k)|\leq\bar{d},\ \forall k,j\geq0$.
\end{assumption}

This assumption is commonly used in ILC or data-driven control to ensure the robustness for resisting the worst-case errors \citep{Rozario_2019}. The physical meaning of the $l_{\infty}$ norm of the disturbance is related to sensor errors, for example an accuracy range of $\pm0.005\Omega$ for a pressure sensor. Then, the Euclidean norm of $d_{j}$ is determined by the data length $N$ and bounded by $\|d_{j}\|_{2}\leq\sqrt{N}\bar{d}$. Notice that only finite $N$ is considered in this work since it is more common in practical applications. The following lemma bounds the maximal error between exact and inexact information.

\begin{lemma}\label{lemma_exact_inexact_bound}
The difference between the pairs $(J(u_{j}),\nabla J(u_{j}))$ and $(\tilde{J}(u_{j}),\tilde{\nabla}J(u_{j}))$ is bounded by following equations:

\begin{small}
\begin{subequations}
\begin{align}
&|J(u_{j})-\tilde{J}(u_{j})|\leq\Delta_{1}, \label{bound_cost_value}\\
&\|\nabla J(u_{j})-\tilde{\nabla} J(u_{j})\|_{2}\leq\Delta_{2} \label{bound_gradient}
\end{align}
\end{subequations}
\end{small}and $\Delta_{1}$ and $\Delta_{2}$ are defined as

\begin{small}
\begin{subequations}
\begin{align}
&\Delta_{1} = \sqrt{2\bar{J}N}\bar{d}+\frac{1}{2}N\bar{d}^{2},\\
&\Delta_{2} =
\begin{Vmatrix}
\tilde{G}-G
\end{Vmatrix}_{2}
\sqrt{2\bar{J}}
+
\begin{Vmatrix}
\tilde{G}
\end{Vmatrix}_{2}
\sqrt{N}\bar{d}
\end{align}
\end{subequations}
\end{small}where $\bar{J}=\max\limits_{u\in\mathcal{Q}}J(u)$.

\end{lemma}

\begin{proof}
Consider the function value in (\ref{exact_information}) and (\ref{inexact_information}), the difference is

\begin{small}
\begin{equation*}
\begin{split}
J(u)-\tilde{J}(u)
&=
\frac{1}{2}\|r-(Gu+c)\|_{2}^{2}
-
\frac{1}{2}\|r-(Gu+c+d)\|_{2}^{2}
\\
&=
\langle
r-(Gu+c),d
\rangle
-\frac{1}{2}\|d\|_{2}^{2}
\end{split}
\end{equation*}
\end{small}Since $\mathcal{Q}$ is a convex and compact set, and $J(u)$ is convex, there exists a finite value $\bar{J}$ such that $\bar{J}=\max\limits_{u\in\mathcal{Q}}J(u)$. Then according to the Cauchy–Schwarz inequality, we have

\begin{small}
\begin{equation*}
\begin{split}
|J(u)-\tilde{J}(u)|
&\leq
\begin{vmatrix}
\langle
r-(Gu+c),d
\rangle
\end{vmatrix}
+\frac{1}{2}N\bar{d}^{2}
\\
&\leq
\|r-(Gu+c)\|_{2}\|d\|_{2}
+\frac{1}{2}N\bar{d}^{2}
\\
&\leq
\sqrt{2\bar{J}N}\bar{d}+\frac{1}{2}N\bar{d}^{2}
\end{split}
\end{equation*}
\end{small}Then consider the difference of the gradient in (\ref{exact_information}) and (\ref{inexact_information}):

\begin{small}
\begin{equation*}
\begin{split}
\|\nabla J(u)-\tilde{\nabla}J(u)\|_{2}
&=
\begin{Vmatrix}
\begin{pmatrix}
\tilde{G}^{T}-G^{T}
\end{pmatrix}
\begin{pmatrix}
r-(Gu+c)
\end{pmatrix}
-\tilde{G}^{T}d
\end{Vmatrix}_{2}
\\
&\leq
\begin{Vmatrix}
\begin{pmatrix}
\tilde{G}^{T}-G^{T}
\end{pmatrix}
\begin{pmatrix}
r-(Gu+c)
\end{pmatrix}
\end{Vmatrix}_{2}
+
\begin{Vmatrix}
\tilde{G}^{T}d
\end{Vmatrix}_{2}
\\
&\leq
\begin{Vmatrix}
\tilde{G}^{T}-G^{T}
\end{Vmatrix}_{2}
\begin{Vmatrix}
r-(Gu+c)
\end{Vmatrix}_{2}
+
\begin{Vmatrix}
\tilde{G}^{T}d
\end{Vmatrix}_{2}
\\
&\leq
\begin{Vmatrix}
\tilde{G}-G
\end{Vmatrix}_{2}
\sqrt{2\bar{J}}
+
\begin{Vmatrix}
\tilde{G}
\end{Vmatrix}_{2}
\sqrt{N}\bar{d}
\end{split}
\end{equation*}
\end{small}In this way, the bounds in (\ref{bound_cost_value}) and (\ref{bound_gradient}) are obtained.
\end{proof}

\begin{remark}
After the system identification, the inexact representation $\tilde{G}$ is fixed, which implies that the bounds in Lemma \ref{lemma_exact_inexact_bound} are time-invariant. Moreover, $\Delta_{1}$ and $\Delta_{2}$ do not have to be known, they only appear in the convergence proof.
\end{remark}

Applying the inexact gradient in (\ref{original_projection_gradient}), we have the following iterative algorithm:

\begin{small}
\begin{subequations}\label{classical_gradient_ILC}
\begin{align}
u_{j+1}
&=
\arg\min\limits_{u\in\mathcal{Q}}
\{
\langle
\tilde{\nabla}J(u_{j}),u-u_{j}
\rangle
+
\frac{L}{2}\|u-u_{j}\|_{2}^{2}
\}
\\
\Rightarrow
u_{j+1}
&=
\Pi_{\mathcal{Q}}
\begin{pmatrix}
u_{j}
+
\frac{1}{L}\tilde{G}^{T}e_{j}
\end{pmatrix}
\end{align}
\end{subequations}
\end{small}in which we assume that the Lipschitz constant $L$ of $J(u_{j})$ is known for simplicity, the iteration can be started from a sufficient small step-size $\rho>0$ if $L$ is unknown. The ILC update law (\ref{classical_gradient_ILC}) can be summarized in Algorithm \ref{algorithm_ILC_classical} and the convergence is shown in Theorem \ref{theorem_convergence_classical}.

\begin{algorithm}
\caption{Classical data-driven iterative learning control}
\label{algorithm_ILC_classical}
\begin{algorithmic}[1]
\REQUIRE~~Data-driven representation $\tilde{G}$, Lipschitz constant $L$, reference trajectory $r$, initial trajectory $u_{0}$ and maximal iteration number $M$.\\
\ENSURE~~ILC sequence $\{u_{j}\}$.\\
\FOR {$j=0:M$}
\STATE Measure the perturbed system output $\tilde{y}_{j}$.
\STATE Compute the tracking error $e_{j}=r-\tilde{y}_{j}$.
\STATE Compute the inexact gradient $\tilde{\nabla} J(u_{j})=-\tilde{G}^{T}e_{j}$.
\STATE Compute \\ $u_{j+1}=\arg\min\limits_{u\in\mathcal{Q}}\{\langle\tilde{\nabla}J(u_{j}),u-u_{j}\rangle+\frac{L}{2}\|u-u_{j}\|_{2}^{2}\}$
\STATE Apply $u_{j+1}$ to the system.
\ENDFOR
\end{algorithmic}
\end{algorithm}

\begin{theorem}\label{theorem_convergence_classical}
Consider the classical data-driven ILC in Algorithm \ref{algorithm_ILC_classical} with Assumptions \ref{assumption_invertible} and \ref{assumption_disturbance}, the following convergence result holds for ILC sequence $\{u_{j}\}$:

\begin{small}
\begin{equation}\label{convergence_rate_classical}
J(u_{j})
-
J(u^{*})
\leq
\frac{L\|u^{*}-u_{0}\|_{2}^{2}}{4j}
+
\delta
\end{equation}
\end{small}where $\delta=2\Delta_{1}+2\Delta_{2}D$ and $D=\sup\limits_{a,b\in\mathcal{Q}}\|a-b\|_{2}$.

\end{theorem}

\begin{proof}
The proof follows the result in \cite{Devolder_2014}. First, we construct the inexact oracle $(J_{\delta,L}(v),\nabla J_{\delta,L}(v))$ from the inexact information $(\tilde{J}(v),\tilde{\nabla}J(v))$. Specifically, as for the global lower bound, we have

\begin{small}
\begin{equation*}
\begin{split}
&J(u)
\geq
J(v)+\langle\nabla J(v),u-v\rangle \\
&\overset{(\ref{bound_cost_value}),(\ref{bound_gradient})}
{\geq}
\tilde{J}(v)-\Delta_{1}
+\langle\tilde{\nabla} J(v),u-v\rangle-\Delta_{2}D
\end{split}
\end{equation*}
\end{small}then define the pair

\begin{small}
\begin{equation*}
\begin{cases}
  J_{\delta,L}(v)=\tilde{J}(v)-\Delta_{1}-\Delta_{2}D & \\
  \nabla J_{\delta,L}(v)=\tilde{\nabla}J(v) &
\end{cases}
\end{equation*}
\end{small}as the inexact oracle. On the other hand, the quadratic upper bound under the inexact oracle is

\begin{small}
\begin{equation*}
\begin{split}
&J(u)
\leq
J(v)+\langle\nabla J(v),u-v\rangle+\frac{L}{2}\|u-v\|_{2}^{2} \\
&\overset{(\ref{bound_cost_value}),(\ref{bound_gradient})}{\leq}
\tilde{J}(v)+\Delta_{1}
+
\langle\tilde{\nabla}J(v),u-v\rangle
+
\Delta_{2}D+\frac{L}{2}\|u-v\|_{2}^{2} \\
&=
J_{\delta,L}(v)
+
\langle
\nabla J_{\delta,L}(v)
,u-v\rangle
+
\frac{L}{2}\|u-v\|_{2}^{2}
+
2\Delta_{1}+2\Delta_{2}D
\end{split}
\end{equation*}
\end{small}and define $\delta=2\Delta_{1}+2\Delta_{2}D$ as the oracle accuracy. In this way, it is easy to get that the inexact oracle satisfies

\begin{small}
\begin{equation}\label{inexact_oracle}
0\leq
J(u)-
\begin{pmatrix}
J_{\delta,L}(v)+\langle\nabla J_{\delta,L}(v),u-v\rangle
\end{pmatrix}
\leq
\frac{L}{2}\|u-v\|_{2}^{2}+\delta,
\ \
\forall u,v\in\mathcal{Q}
\end{equation}
\end{small}Then according to Theorem 2 in \cite{Devolder_2014}, we have (\ref{convergence_rate_classical}).
\end{proof}

Moreover, we use the Nesterov accelerated gradient method and the inexact oracle to solve the ILC problem, which is summarized in Algorithm \ref{algorithm_ILC} and the convergence result is given in Theorem \ref{theorem_convergence}.

\begin{algorithm}
\caption{Fast data-driven iterative learning control}
\label{algorithm_ILC}
\begin{algorithmic}[1]
\REQUIRE~~Data-driven representation $\tilde{G}$, Lipschitz constant $L$, reference trajectory $r$, initial trajectory $u_{0}$ and maximal iteration number $M$.\\
\ENSURE~~ILC sequence $\{u_{j}\}$.\\
\FOR {$j=0:M$}
\STATE Measure the perturbed system output $\tilde{y}_{j}$.
\STATE Compute the tracking error $e_{j}=r-\tilde{y}_{j}$.
\STATE Compute the inexact gradient $\tilde{\nabla} J(u_{j})=-\tilde{G}^{T}e_{j}$.
\STATE Compute \\ $\mu_{j}=\arg\min\limits_{\mu\in\mathcal{Q}}\{\langle\tilde{\nabla} J(u_{j}),\mu-u_{j}\rangle+\frac{L}{2}\|\mu-u_{j}\|_{2}^{2}\}$
\STATE Compute \\
$
\nu_{j}=\arg\min\limits_{\nu\in
\mathcal{Q}}\{
\frac{1}{2}\|\nu-u_{0}\|_{2}^{2}
+\sum_{i=0}^{j}\frac{i+1}{2L}\langle\tilde{\nabla}J(u_{i}),\nu-u_{i}\rangle\}
$\\
\STATE Compute $u_{j+1}=\frac{2}{j+3}\nu_{j}+\frac{j+1}{j+3}\mu_{j}$
\STATE Apply $u_{j+1}$ to the system.
\ENDFOR
\end{algorithmic}
\end{algorithm}

\begin{theorem}\label{theorem_convergence}
Consider the fast data-driven ILC in Algorithm \ref{algorithm_ILC} with Assumptions \ref{assumption_invertible} and \ref{assumption_disturbance}, the following convergence result holds for ILC sequence $\{u_{j}\}$:

\begin{small}
\begin{equation}\label{convergence_rate}
J(u_{j})
-
J(u^{*})
\leq
\frac{2L\|u^{*}-u_{0}\|_{2}^{2}}{(j+1)(j+2)}
+
\frac{j+3}{3}\delta
\end{equation}
\end{small}where $\delta$ and $D$ are same as in Theorem \ref{theorem_convergence_classical}.

\end{theorem}

\begin{proof}
Similar to the proof in Theorem \ref{theorem_convergence_classical}, according to Theorem 4 in \cite{Devolder_2014}, we have

\begin{small}
\begin{equation*}
\begin{split}
\frac{(j+1)(j+2)}{4L}J(u_{j})
&\leq
\min\limits_{\nu\in\mathcal{Q}}
\begin{Bmatrix}
\frac{1}{2}\|\nu-u_{0}\|_{2}^{2}
+
\sum_{i=0}^{j}\frac{i+1}{2L}
\begin{pmatrix}
J_{\delta,L}(u_{i})
+
\langle
\nabla J_{\delta,L}(u_{i}),\nu-u_{i}
\rangle
\end{pmatrix}
\end{Bmatrix}\\
&+
\delta\sum_{i=0}^{j}\frac{(i+1)(i+2)}{4L}
\end{split}
\end{equation*}
\end{small}The minimization in the right hand side follows the inequality:

\begin{small}
\begin{equation*}
\begin{split}
&\min\limits_{\nu\in\mathcal{Q}}
\begin{Bmatrix}
\frac{1}{2}\|\nu-u_{0}\|_{2}^{2}
+
\sum_{i=0}^{j}\frac{i+1}{2L}
\begin{pmatrix}
J_{\delta,L}(u_{i})
+
\langle
\nabla J_{\delta,L}(u_{i}),\nu-u_{i}
\rangle
\end{pmatrix}
\end{Bmatrix}\\
&\leq
\frac{1}{2}\|u^{*}-u_{0}\|_{2}^{2}
+
\sum_{i=0}^{j}\frac{i+1}{2L}
\begin{pmatrix}
J_{\delta,L}(u_{i})
+
\langle
\nabla J_{\delta,L}(u_{i}),u^{*}-u_{i}
\rangle
\end{pmatrix}\\
&
\overset{(\ref{inexact_oracle})}
{\leq}
\frac{1}{2}\|u^{*}-u_{0}\|_{2}^{2}
+
\sum_{i=0}^{j}\frac{i+1}{2L}
J(u^{*})
\\
&
=
\frac{1}{2}\|u^{*}-u_{0}\|_{2}^{2}
+
\frac{(j+1)(j+2)}{4L}J(u^{*})
\end{split}
\end{equation*}
\end{small}Then one obtains the result:

\begin{small}
\begin{equation*}
\frac{(j+1)(j+2)}{4L}J(u_{j})
\leq
\frac{1}{2}\|u^{*}-u_{0}\|_{2}^{2}
+
\frac{(j+1)(j+2)}{4L}J(u^{*})
+
\delta\sum_{i=0}^{j}\frac{(i+1)(i+2)}{4L}
\end{equation*}
\end{small}which leads to (\ref{convergence_rate}).
\end{proof}

The convergence result (\ref{convergence_rate_classical}) shows that there is no error accumulation and the upper bound of $J(u_{j})-J(u^{*})$ converges to $\delta$ with the rate $\mathcal{O}(1/j)$. On the other hand, the first term in (\ref{convergence_rate}) decreases with $\mathcal{O}(1/j^{2})$, but the second term increases with $\mathcal{O}(j)$, which implies that fast ILC is asymptotically divergent. However, the divergent term can only dominate the convergence rate after a few iterations, therefore, fast ILC is superior to classical ILC at the early iteration stage.

Based on this observation, a natural way to improve the convergence performance is to combine fast and classical ILC, which is called hybrid ILC in this paper. Specifically, fast ILC is used first in order to decrease the tracking error rapidly, and then the algorithm is switched to classical ILC for a better asymptotic performance. Assume the algorithm is switched at the iteration $s$, it is easy to see that hybrid ILC has the same convergence result as fast ILC when $j\leq s$, and if $j>s$, the classical ILC convergence result initialized from $u_{s}$ is inherited. There are two merits of applying this hybrid method. First, classical ILC can start from $u_{s}$ early. Classical ILC usually takes more iterations to calculate the same input $u_{s}$ compared to fast ILC. Second, the divergence effect of fast ILC is removed. The input $u_{s}$ can be improved further by classical ILC, rather than oscillating due to the error accumulation.

The switching iteration $s$ is determined by some switching conditions, which is detected at each iteration. In this paper, we propose an empirical switching condition, which is simple but efficient in experiments. Particularly, we calculate the average tracking error norm $\bar{e}_{j}$ at the $j$-th iteration within the time window $[j-W+1,j]$, in which the window length $W$ is determined by the user, and then compare with the previous average tracking error norm $\bar{e}_{j-1}$. If $\bar{e}_{j}>\bar{e}_{j-1}$, then the algorithm switches to classical ILC. The proposed hybrid ILC method is summarized in Algorithm \ref{algorithm_hybrid_ILC}. In line $13$, fast ILC is used to decrease the tracking error rapidly at the early iteration stage. Lines $5-11$ summarize the detection of switching condition. Line $15$ shows that classical ILC is used for a better asymptotic performance after the switching.

The following theorem provides the theoretical guarantee of hybrid ILC to achieve a given accuracy level $\tau$:

\begin{algorithm}[t]
\caption{Hybrid data-driven iterative learning control}
\label{algorithm_hybrid_ILC}
\begin{algorithmic}[1]
\REQUIRE~~Data-driven representation $\tilde{G}$, Lipschitz constant $L$, reference trajectory $r$, initial trajectory $u_{0}$, maximal iteration number $M$, moving horizon length $W$ and initial switching flag $SF=0$.\\
\ENSURE~~ILC sequence $\{u_{j}\}$.\\
\FOR {$j=0:M$}
\STATE Measure the perturbed system output $\tilde{y}_{j}$.
\STATE Compute the tracking error $e_{j}=r-\tilde{y}_{j}$.
\STATE Compute the inexact gradient $\tilde{\nabla} J(u_{j})=-\tilde{G}^{T}e_{j}$.
\IF{$j\geq W$ and $SF=0$}
\STATE Compute the averaged tracking error within the current time window $[j-W+1,j]$: \\ $\bar{e}_{j}=\frac{1}{W}\sum_{i=j-W+1}^{j}\|e_{i}\|_{2}$
\STATE Compute the averaged tracking error within the  previous time window $[j-W,j-1]$: \\ $\bar{e}_{j-1}=\frac{1}{W}\sum_{i=j-W}^{j-1}\|e_{i}\|_{2}$
\IF{$\bar{e}_{j}>\bar{e}_{j-1}$}
\STATE Set $SF=1$.
\ENDIF
\ENDIF
\IF{$SF=0$}
\STATE Invoke steps $5-7$ in Algorithm \ref{algorithm_ILC} (Fast ILC).
\ELSIF{$SF=1$}
\STATE Invoke step $5$ in Algorithm \ref{algorithm_ILC_classical} (Classical ILC).
\ENDIF
\STATE Apply $u_{j+1}$ to the system.
\ENDFOR
\end{algorithmic}
\end{algorithm}

\begin{theorem}\label{theorem_iteration_bound}
Consider the hybrid data-driven ILC in Algorithm \ref{algorithm_hybrid_ILC} with Assumptions \ref{assumption_invertible} and \ref{assumption_disturbance}. If the algorithm is switched at the iteration $s$ and the accuracy level is given by $\tau>\delta$, then the ILC sequence $\{u_{j}\}$ satisfies $J(u_{j})-J(u^{*})\leq\tau$ for all $j\geq \lceil j^{*}\rceil$, and $j^{*}$ is formulated by

\begin{small}
\begin{equation}\label{iteration_bound}
j^{*}=
\frac{6L^{2}\|u^{*}-u_{0}\|_{2}^{2}+\delta L(s+1)(s+2)(s+3)}
{6\xi(s+1)(s+2)(\tau-\delta)}
+s
\end{equation}
\end{small}where $\xi=\lambda_{\min}(G^{T}G)$ and $\lambda_{\min}(\cdot)$ is the minimum eigenvalue.

\end{theorem}

\begin{proof}
According to Assumption \ref{assumption_invertible} and the lower triangular structure of $G$, we have $G^{T}G\succ0$, hence the exact cost function $J(u)=\frac{1}{2}\|r-(Gu+c)\|_{2}^{2}$ is strongly convex with the parameter $\xi$. Thus, $J(u)$ has a quadratic lower bound:

\begin{small}
\begin{equation*}
J(u)
\geq
J(v)+\nabla J(v)^{T}(u-v)+\frac{\xi}{2}\|u-v\|_{2}^{2},\
\forall u,v\in\mathcal{Q}
\end{equation*}
\end{small}Let $u=u_{j}$ and $v=u^{*}$, and notice that $\nabla J(u^{*})=0$, then we have a lower bound of $J(u_{j})-J(u^{*})$:

\begin{small}
\begin{equation*}
J(u_{j})-J(u^{*})
\geq
\frac{\xi}{2}\|u_{j}-u^{*}\|_{2}^{2}
\end{equation*}
\end{small}Since the algorithm is switched at the iteration $s$, (\ref{convergence_rate}) leads to

\begin{small}
\begin{equation*}
\frac{\xi}{2}\|u_{s}-u^{*}\|_{2}^{2}
\leq
J(u_{s})
-
J(u^{*})
\leq
\frac{2L\|u^{*}-u_{0}\|_{2}^{2}}{(s+1)(s+2)}
+
\frac{s+3}{3}\delta
\end{equation*}
\end{small}On the other hand, (\ref{convergence_rate_classical}) is inherited for all $j>s$:

\begin{small}
\begin{equation*}
J(u_{j})
-
J(u^{*})
\leq
\frac{L\|u^{*}-u_{s}\|_{2}^{2}}{4(j-s)}
+
\delta
\end{equation*}
\end{small}which leads to

\begin{small}
\begin{equation*}
J(u_{j})
-
J(u^{*})
\leq
\frac{6L^{2}\|u^{*}-u_{0}\|_{2}^{2}+\delta L(s+1)(s+2)(s+3)}
{6\xi(s+1)(s+2)(j-s)}
+
\delta
\end{equation*}
\end{small}Let the right hand side of the above inequality be less than or equal to $\tau$, then we have (\ref{iteration_bound}).
\end{proof}

\begin{remark}
Theorem \ref{theorem_iteration_bound} assumes that the accuracy level $\tau$ is achieved after the algorithm switching, i.e., $j^{*}>s$, this is practical since convergence rates (\ref{convergence_rate_classical}) and (\ref{convergence_rate}) imply that classical ILC has a better asymptotic performance, therefore, a higher accuracy can be achieved. Moreover, (\ref{convergence_rate_classical}) shows that the upper bound of the cost function accuracy decreases with the iteration number $j$ and converges to $\delta$ asymptotically, therefore, the given accuracy level $\tau$ has to be larger than $\delta$.
\end{remark}

\begin{remark}
The switching condition of hybrid ILC is determined by the user, Algorithm \ref{algorithm_hybrid_ILC} only provides a simple one. The essence is to observe the tracking error $e_{j}$, if $e_{j}$ becomes oscillating or starts to diverge, for example around step $50$ in Figure \ref{fig_tracking_input_error}(a), then one should switch to classical ILC in order to reduce the tracking error further.
\end{remark}

\section{Illustrative examples}\label{sec5}

\subsection{Toy example}

In this part, the performance of different methods are evaluated on the system:

\begin{small}
\begin{equation*}
H(z)=
\frac{0.7836z^{3}+0.7732z^{2}+0.1936z+0.009937}
{z^{4}+1.778z^{3}+0.9869z^{2}+0.2007z+0.0205}
\end{equation*}
\end{small}which is controllable, minimum-phase and of relative degree $1$. Each element of the invariant reference signal and initial state $x_{0}$ are randomly generated from the interval $[-10,10]$. The box-constraint set $\mathcal{Q}$ for each input element is selected from $[-20,20]$, with a guarantee that the optimal input trajectory $u^{*}=G^{-1}(r-c)$ is included ($G$ and $c$ are defined in (\ref{convolutional_matrix})), which leads to $J(u^{*})=0$. The proposed method has two parts: data-driven representation (offline part) and iterative learning control (online part). The offline data length is $T=1000$ and the initial trajectory length $T_{\mathrm{ini}}=4$, which is equal to the system lag. The trial length is $N=20$, the maximal iteration number is $M=500$ and the time window length is $W=20$ for the online part.

Three kinds of output disturbances are considered: (a), the uniformly distributed disturbance with the bound $\bar{d}$. (b), the sine disturbance with the bound $\bar{d}$, $100$rad/s angular frequency and random phases. (c), the white Gaussian disturbance with $0$ mean and $\sigma$ standard deviation. The disturbance data length is $T$ for the offline part and $N$ for the online part. The bound is $\bar{d}=3$ for the first two disturbances, and the standard deviation is $\sigma=1$ for the last disturbance. Although the white Gaussian disturbance can be unbounded, around $99.73\%$ of samples belong to the interval $[-3\sigma,3\sigma]$ according to the three-$\sigma$ rule.

In order to evaluate the offline performance, we discuss the relative error generated by the data-driven representations $\tilde{G}_{\mathrm{DD}}$ and $\tilde{G}_{\mathrm{SI}}$, in which DD and SI denote the non-parametric direct data-driven method obtained by the fundamental lemma and the parametric representation obtained by subspace identification, respectively. In particular, a feasible input trajectory $u$ is randomly generated after the offline part, and then used to calculate the relative error $\|y-\tilde{y}\|_{2}/\|y\|_{2}$ by $y=Gu$ (exact output), and $\tilde{y}_{\mathrm{DD}}=\tilde{G}_{\mathrm{DD}}u$ or $\tilde{y}_{\mathrm{SI}}=\tilde{G}_{\mathrm{SI}}u$, in which the initial response is assumed to be $0$. As for the online performance, we use the exact tracking error, i.e., $\varepsilon_{j}$ in (\ref{exact_information}), and the input error $u_{j}-u^{*}$ as the indicators.

\begin{figure}[h]
\centering
\includegraphics[scale=0.5]{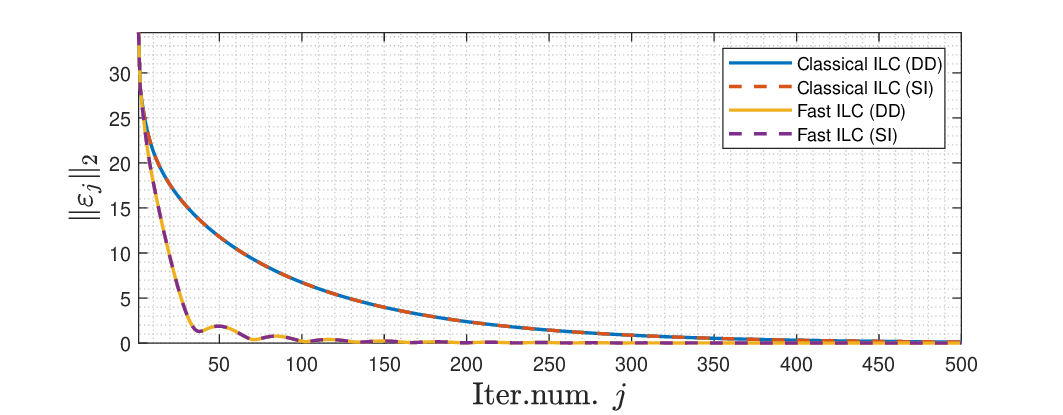}
\caption{Tracking error in the disturbance-free case.}
\label{fig_disturbance_free}
\end{figure}

In the disturbance-free case, i.e., $\bar{d}=0$, we have $\tilde{G}_{\mathrm{DD}}=\tilde{G}_{\mathrm{SI}}=G$ and $\Delta_{1}=\Delta_{2}=\delta=0$, therefore, the relative errors of the two data-driven representations are $0$. Applying such $\tilde{G}_{\mathrm{DD}}$ and $\tilde{G}_{\mathrm{SI}}$ to the online control, the tracking error $\|\varepsilon_{j}\|_{2}$ is depicted in Figure \ref{fig_disturbance_free}, which shows that fast ILC converges much faster than classical ILC. Notice that hybrid ILC should not be used in the disturbance-free case since there is no divergence term in (\ref{convergence_rate}). In the disturbance case, we set the signal-to-noise ratio (SNR) of the offline data to $10$dB . The relative errors are shown in Table \ref{table_RE}, which implies that the direct data-driven method yields larger relative errors compared to the system identification method in the cases of uniform and white Gaussian disturbance. This is to be expected, given that subspace identification is designed to exploit the full knowledge of the state-space representation behind the data with white Gaussian disturbance \citep{Overschee_2012}. On the other hand, the direct data-driven method shows better performance in the case of structured disturbance, exhibiting much smaller relative errors compared to SI. One possible explanation is that the non-parametric nature of the direct data-driven method is able to better capture the additional system dynamics implicitly responsible for the structured disturbance.

\begin{table}[h]
\centering
\caption{Relative error of two representation methods}
\begin{tabular}{lccc}
\toprule
   & Uniform & Sine & Gaussian \\ \hline
DD & 0.1843 & \textbf{2.5786e-05}  & 0.1887  \\ \hline
SI & \textbf{0.0392} & 0.0413  & \textbf{0.0145}  \\
\bottomrule
\end{tabular}\label{table_RE}
\end{table}

\begin{figure}
\centering
\subfloat[Uniformly distributed random disturbance.]{%
\resizebox*{7cm}{!}{\includegraphics[width=3.6in]{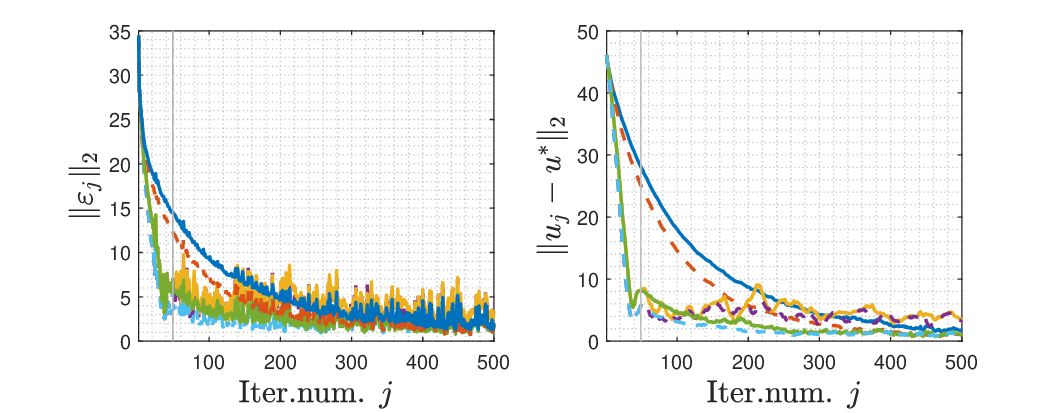}}}\hspace{5pt}
\subfloat[Sine disturbance.]{%
\resizebox*{7cm}{!}{\includegraphics[width=3.6in]{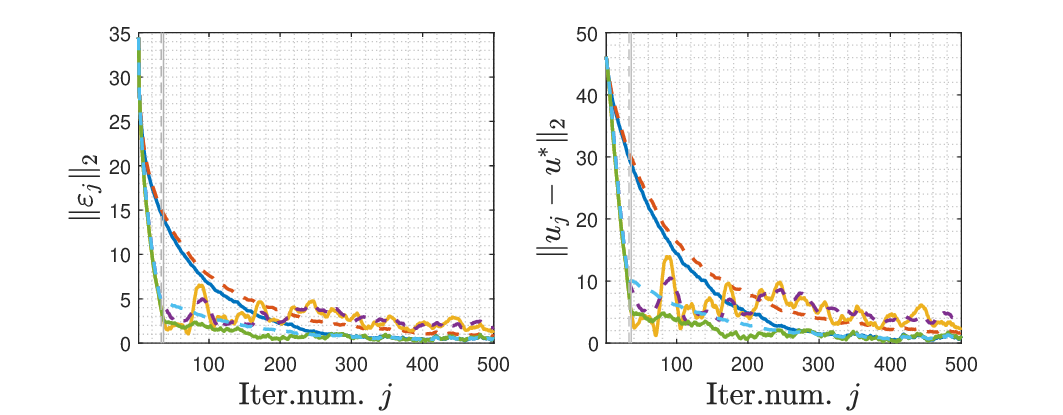}}}\hspace{5pt}
\subfloat[White Gaussian disturbance.]{%
\resizebox*{7cm}{!}{\includegraphics[width=3.6in]{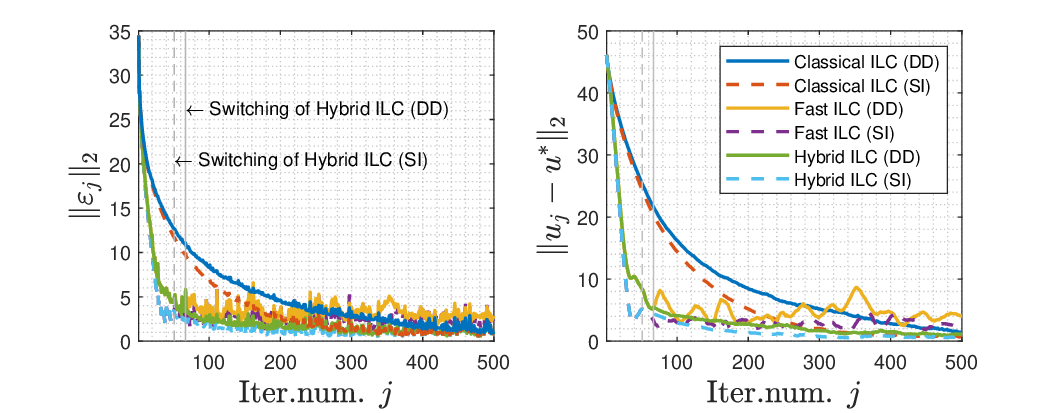}}}
\caption{Tracking and input errors in the disturbance case.} \label{fig_tracking_input_error}
\end{figure}

Next, we apply $\tilde{G}_{\mathrm{DD}}$ and $\tilde{G}_{\mathrm{SI}}$ to the above three ILC methods: classical, fast and hybrid ILC. The tracking error $\|\varepsilon_{j}\|_{2}$ and input error $\|u_{j}-u^{*}\|_{2}$ are depicted in Figure \ref{fig_tracking_input_error}, which indicates that hybrid ILC has the same fast convergence speed as fast ILC at the early iteration stage, and has an excellent asymptotic performance after switching. In the other words, hybrid ILC achieves equally good asymptotic performance as classical ILC using $100-200$ iterations less. However, in the cases of uniform and white Gaussian disturbance, each ILC method with the non-parametric representation cannot outperform the parametric representation. As discussed above, the direct data-driven method has a limitation in handling the unstructured disturbance. On the other hand, as expected, the ILC method with the non-parametric representation outperforms the parametric representation in the case of sine disturbance because of the high representation accuracy, see Table \ref{table_RE}. In this way, the effectiveness of the non-parametric data-driven representation and hybrid ILC are illustrated.

\subsection{Batch of random systems for the structured disturbance}

The above toy example shows that the direct data-driven method is promising when it comes to handling structured disturbances. In this subsection, we test the different methods on a batch of randomly generated systems affected by sine disturbance, and discuss the statistical performance. We use $100$ random systems of order $n=4$, generated using the MATLAB function $drss$. All systems are controllable, minimum-phase and of relative degree $1$. The disturbance setting is the same as in the previous subsection. The input box-constraint set $\mathcal{Q}$ and the reference $r$ are randomly generated for each system.

The average relative errors yielded by the non-parametric and parametric representations are  $0.0261$ and $0.0636$, which implies that the direct data-driven method can obtain smaller relative errors in most of experiments compared to subspace identification method. Figure \ref{fig_ave_sine_3}(a) shows that three ILC methods equipped with the non-parametric representation are more robust to the sine disturbance on average when it has a large bound $\bar{d}$, and then a better control performance can be achieved. On the other hand, note that the early iteration stage of an ILC process can reflect the convergence speed and the final stage can reflect the steady-state tracking error, we discuss the performance of different ILC methods at these two stages in the following. We select the average tracking errors at the two stages as indicators and the stage data length as $50$, then the early stage performance indicator is $\frac{1}{50}\sum_{j=1}^{50}\|\varepsilon_{j}\|_{2}$ and the final stage performance indicator is $\frac{1}{50}\sum_{j=451}^{500}\|\varepsilon_{j}\|_{2}$. A smaller indicator value implies a better performance: faster convergence speed or smaller steady-state tracking error. The box plots that counted all the random experiments are given as Figure \ref{fig_ave_sine_3}(b). The left box plot shows that hybrid ILC has better statistical performance at the early stage. Compared to the parametric data-driven representation, the ILC methods with non-parametric representation yield smaller errors, which confirm that the non-parametric representation is more robust to the large structured disturbance. Since fast ILC is sensitive to disturbances (tracking errors may oscillate significantly, see Figure \ref{fig_tracking_input_error}(b)), the $0$-th, $25$-th and $50$-th percentiles of fast ILC are larger than classical ILC. This drawback is improved by hybrid ILC. The right box plot in Figure \ref{fig_ave_sine_3}(b) implies that smaller tracking errors at the final stage can be obtained by hybrid ILC, the reason is that hybrid ILC achieves a similar asymptotic performance to classical ILC in advance and then the tracking error can be reduced further. It is worth noting that the tracking error may be decreasing very slowly if only using classical ILC (the $75$-th and $100$-th percentiles can be even larger than fast ILC), which implies the necessity of a switching mechanism.

\begin{figure}[h]
\centering
\subfloat[Average tracking errors of $100$ random experiments.]{%
\resizebox*{7cm}{!}{\includegraphics[width=3.2in]{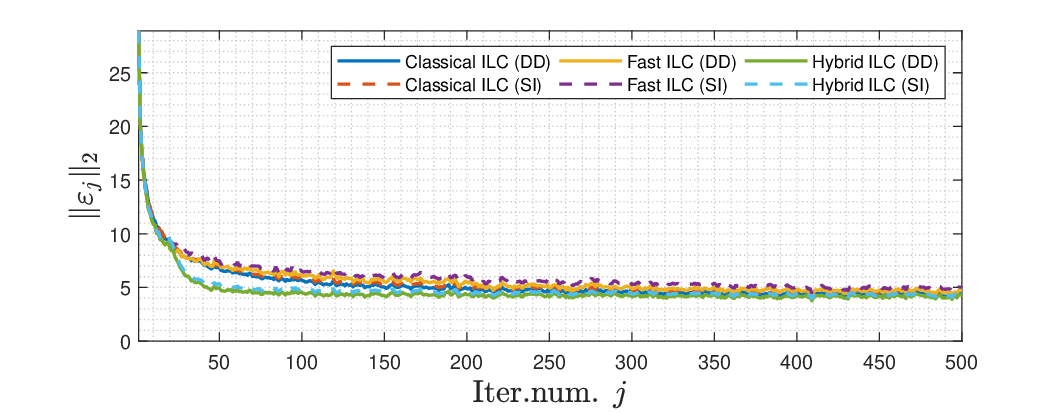}}}\hspace{5pt}
\subfloat[Average tracking errors of the first and last $50$ steps for each random experiment (C: Classical, F: Fast, H: Hybrid).]{%
\resizebox*{7cm}{!}{\includegraphics[width=3.2in]{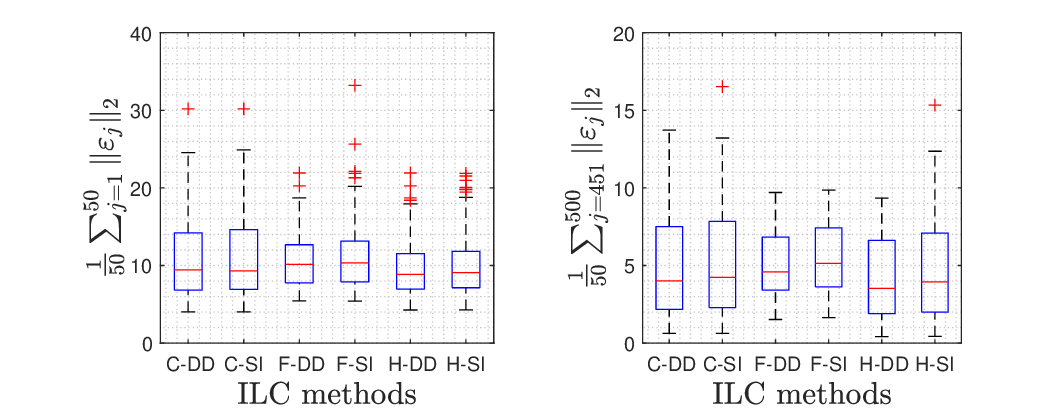}}}
\caption{Performance of different ILC methods with the sine disturbance in the case of $\bar{d}=3$.} \label{fig_ave_sine_3}
\end{figure}

The algorithm performance on a small disturbance bound is also discussed. We choose $\bar{d}=0.5$, the non-parametric and parametric representations yield similar average relative errors as $0.0192$ and $0.0211$, respectively. Figure \ref{fig_ave_sine_0.5}(a) implies that each ILC method can provide a similar performance by using the two data-driven representations. Figure \ref{fig_ave_sine_0.5}(b) also shows the similarity of the two representations in the small structured disturbance case. Furthermore, the advantages of hybrid ILC are indicated continually. Combined the result in Figures \ref{fig_ave_sine_3} and \ref{fig_ave_sine_0.5}, we conclude that the non-parametric representation is more robust to the large structured disturbance on average and performs similarly to the parametric representation with the small structured disturbance. We also suggest that hybrid ILC should always be used, which can achieve a faster tracking error decrease at the early stage and obtain a better asymptotic performance after the switching.

\begin{figure}[h]
\centering
\subfloat[Average tracking errors of $100$ random experiments.]{%
\resizebox*{7cm}{!}{\includegraphics[width=3.2in]{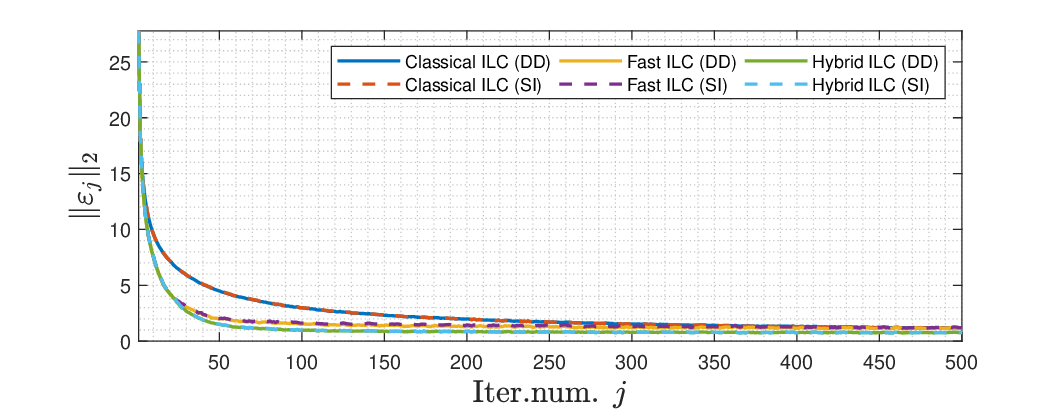}}}\hspace{5pt}
\subfloat[Average tracking errors of the first and last $50$ steps for each random experiment (C: Classical, F: Fast, H: Hybrid).]{%
\resizebox*{7cm}{!}{\includegraphics[width=3.2in]{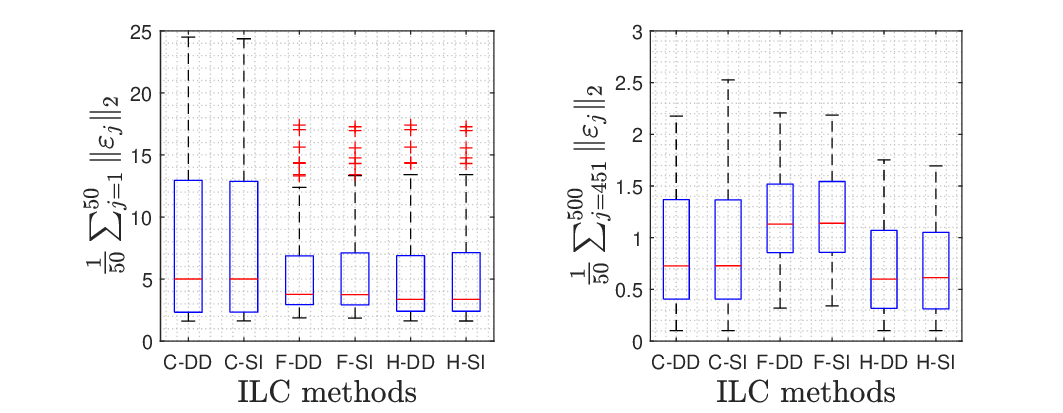}}}
\caption{Performance of different ILC methods with the sine disturbance in the case of $\bar{d}=0.5$.} \label{fig_ave_sine_0.5}
\end{figure}

\begin{remark}
The discussion in this subsection is based on the statistical performance, there also exist some systems that the non-parametric representation cannot outperform the parametric one. The system structure under which direct data-driven method can provide a superior performance with the structured disturbance will be studied in the future work.
\end{remark}

\begin{remark}
The comparison between the non-parametric and parametric representations is based on the same order $n$, i.e., using $n$ in (\ref{svd_n}) and the MATLAB function $n4sid$. The essence of structured disturbances is an additional system dynamics of order $m$, hence, using the order $n+m$ in the function $n4sid$ can perform equally well if not better than the direct data-driven method. In practice, however, the additional order $m$ is usually unknown, but the system order $n$ is known a priori, for example the robotic motion system in the next subsection, therefore, the direct data-driven method can be more practical in some applications.
\end{remark}

\subsection{Case study on the robotic motion system}

In this subsection, we provide a case study on a high-precision two-axis robotic motion system, see \cite{Bristow_2006_2}. Since the X and Y axes are decoupled, we only consider the Y-axis for simplicity and the feedforward control of the X-axis can be deployed in a similar way. The Y-axis model is

\begin{small}
\begin{equation*}
G_{Y}(s)
=
\frac{3500\begin{pmatrix}\frac{1}{88.5^{2}}s^{2}+\frac{0.1}{88.5}s+1\end{pmatrix}}
{s^{2}\begin{pmatrix}\frac{1}{89.5^{2}}s^{2}+\frac{0.1}{89.5}s+1\end{pmatrix}}
\ \
\begin{pmatrix}
\frac{\text{mm}}{\text{V}}
\end{pmatrix}
\end{equation*}
\end{small}and the feedback controller is

\begin{small}
\begin{equation*}
C_{Y}(s)
=
\frac{10(s+1.9)(s+100)}
{(s+1.15)(s+200)}
\ \
\begin{pmatrix}
\frac{\text{V}}{\text{mm}}
\end{pmatrix}
\end{equation*}
\end{small}in which the tracking error dynamics is stabilized with a feedback controller $C_{Y}(s)$ and ILC is used to generate the feedforward input $u$. The control structure is shown in Figure \ref{fig_control_structure}.

\begin{figure}[h]
\centering
\resizebox{250pt}{60pt}{%
\begin{tikzpicture}
[
     node distance = 10mm and 10mm,
        box/.style = {draw, text width=18mm, minimum height=8mm, align=center},
pics/adjbox/.style = {
              code = {\node (@adjbox) [box] {#1};
                     \coordinate[below=3mm of @adjbox.south west] (-adjb);
                     \coordinate[above=3mm of @adjbox.north east] (adjt);
                     \draw[->] (-adjb) -- (adjt);
                     \coordinate (-in1) at ([yshift=+3mm] @adjbox.west);
                     \coordinate (-in2) at ([yshift=-3mm] @adjbox.west);
                     \coordinate (-out) at (@adjbox.east);}
                     },
        sum/.style = {circle, draw, node contents={}},
                > = Stealth
]

\node (n_controller) [ultra thick][box] {$C_{Y}$};
\coordinate[right=of n_controller] (aux1);

\node (n_plant) [ultra thick][box,right =of aux1] {$G_{Y}$};

\node (n_sum_1) [ultra thick][sum,right=of n_controller];
\node (n_sum_2) [ultra thick][sum,left=of n_controller];
\node (n_sum_3) [ultra thick][sum,right=of n_plant];

\draw[ultra thick][->] (n_sum_2) -- (n_controller)node[pos=0.4,above] {$e$};
\draw[ultra thick][->] (n_controller) -- (n_sum_1);
\draw[ultra thick][->] (n_sum_1) -- (n_plant);

\draw[ultra thick][->] (n_plant)--(n_sum_3)node[pos=0.4,above] {$y$};

\coordinate[right=of n_sum_3] (aux_n_sum_3);
\coordinate[right=of aux_n_sum_3] (aux_n_sum_3_right);
\coordinate[below=of n_controller] (aux_n_controller);
\draw[ultra thick][->] (n_sum_3) -- (aux_n_sum_3_right)node[pos=0.5,above] {$\tilde{y}$};
\draw[ultra thick][->] (aux_n_sum_3) |- (aux_n_controller) -| (n_sum_2)node[pos=0.9,right] {$\mathbf{-}$};

\coordinate[left=of n_sum_2] (aux_r);
\draw[ultra thick][->] (aux_r) -- (n_sum_2)node[pos=0.5,above] {$r$};

\coordinate[above=of n_sum_1] (aux_n_sum_1);
\draw[ultra thick][->] (aux_n_sum_1) -- (n_sum_1)node[pos=0.3,right] {$u$}node[pos=0.9,right] {$\mathbf{+}$};

\coordinate[above=of n_sum_3] (aux_n_sum_3);
\draw[ultra thick][->] (aux_n_sum_3) -- (n_sum_3)node[pos=0.3,right]{$d$}node[pos=0.9,right] {$\mathbf{+}$};
\end{tikzpicture}
}
\caption{Control structure of the precision motion system.}
\label{fig_control_structure}
\end{figure}
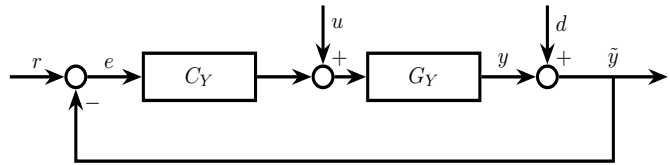

The transfer function from reference $r$ and feedforward $u$ to position $y$ is

\begin{small}
\begin{equation*}
y(s)
=
\frac{G_{Y}}{1+G_{Y}C_{Y}}u(s)
+
\frac{G_{Y}C_{Y}}{1+G_{Y}C_{Y}}r(s)
\end{equation*}
\end{small}in which the sampling rate is set to be $100\text{Hz}$. For the system identification, the data trajectories of $y$, $u$ and $r$ with the length $T=500$ are collected. The length of the initial trajectory is $T_{\mathrm{ini}}=12$, which is equal to the lag of the closed-loop system. The data length in one trial is $N=10$. The box-constraint on $u_{j}(k)$ is $0\text{V}\sim10\text{V}$. To start the iteration, each element in the initial trajectory $u_{0}$ is selected as the center of the box-constraint, i.e., $u_{0}(k)=5$ for $k\in\{0\ldots N-1\}$. Each element of the invariant reference $r$ is randomly generated from $0.77\text{mm}\sim0.81\text{mm}$. The initial condition is randomly generated and invariant in each trial. The maximal iteration number is fixed as $M=20$ and the time window length is $W=5$.

Seven methods are compared: classical, fast and hybrid ILC equipped with the non-parametric and parametric representations, and the feedback controller $C_{Y}$ (without the feedforward). In the disturbance-free case, the tracking error $\|\varepsilon_{j}\|_{2}$ is shown in Figure \ref{fig_case_study_disturbance_free}. Except for the similar conclusion to the above subsections, we also observe that a better tracking performance is obtained by the ILC methods, compared to only using the feedback controller, which illustrates the necessity of the feedforward in real applications.

\begin{figure}[h]
\centering
\includegraphics[scale=0.5]{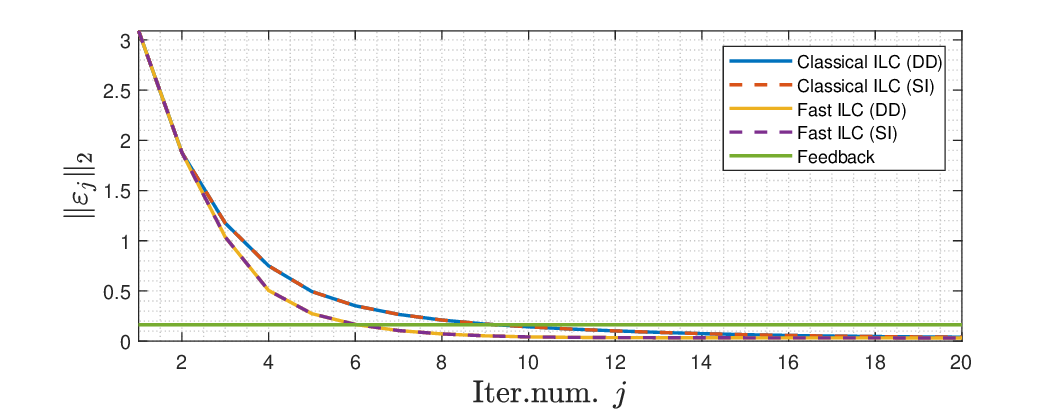}
\caption{Tracking error in the disturbance-free case.}
\label{fig_case_study_disturbance_free}
\end{figure}

\begin{figure}
\centering
\subfloat[Uniformly distributed random disturbance.]{%
\resizebox*{7cm}{!}{\includegraphics[width=3.2in]{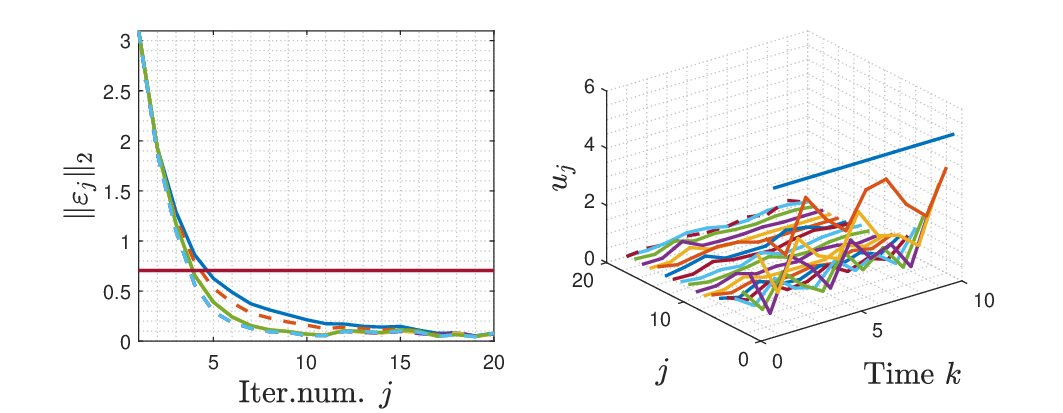}}}\hspace{5pt}
\subfloat[Sine disturbance.]{%
\resizebox*{7cm}{!}{\includegraphics[width=3.2in]{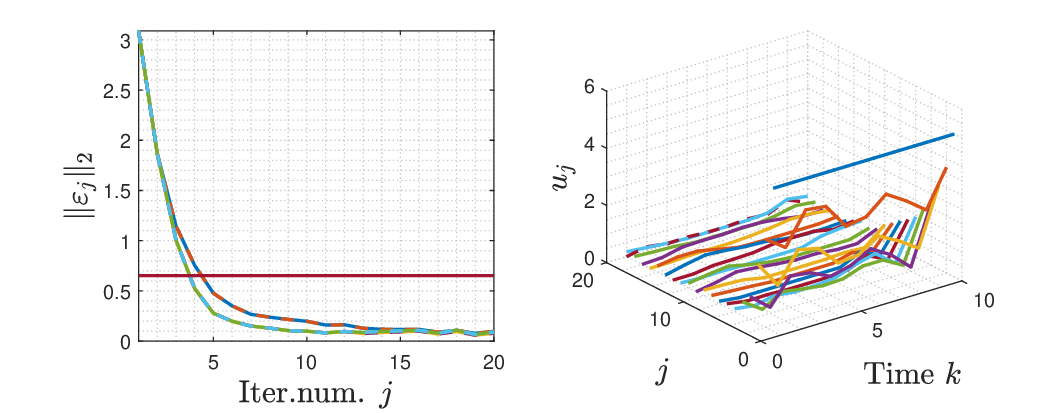}}}\hspace{5pt}
\subfloat[White Gaussian disturbance.]{%
\resizebox*{7cm}{!}{\includegraphics[width=3.2in]{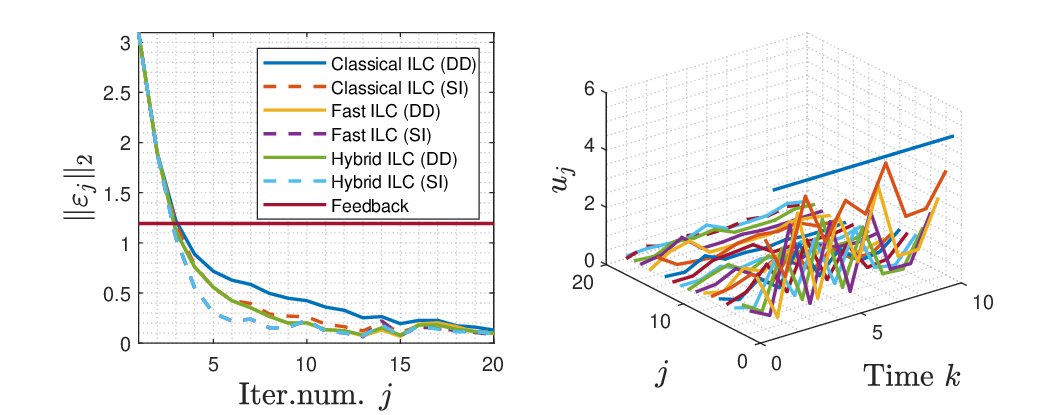}}}
\caption{Tracking error and control convergence in the disturbance case.} \label{fig_case_study_tracking}
\end{figure}

In the disturbance case, since the required control accuracy is high, we choose the disturbance level as $\bar{d}=0.1$ and the standard deviation of white Gaussian disturbance as $\sigma=0.1$. The performance is illustrated by Figure \ref{fig_case_study_tracking}. The left column of Figure \ref{fig_case_study_tracking} is the tracking errors under different disturbances. As expected for the structured disturbance, the non-parametric and parametric representations perform similarly since the disturbance bound is small. In the cases of uniform and white Gaussian disturbance, the parametric representation performs better than the non-parametric one. Moreover, hybrid ILC performs similarly with fast ILC because of the minor disturbance, which leads to a small $\delta$ and benefits the convergence from the first term in (\ref{convergence_rate}). This similarity also implies that hybrid ILC is not worse than classical and fast ILC under the small disturbance. In fact, the performance improvement introduced by hybrid ILC is almost a free lunch, hybrid ILC requires the calculation, comparison and storage of the averaged tracking error norm, and these operations are cheap in the modern digital controller. The right column of Figure \ref{fig_case_study_tracking} shows the control convergence of hybrid ILC with the non-parametric representation: $u_{j}$ converges to $u^{*}$ (dotted line) with the increase of $j$. In this way, Figure \ref{fig_case_study_tracking} confirms the conclusions of previous subsections and the effectiveness of the proposed method is verified.

\section{Conclusion}\label{sec6}

This work presents a fast data-driven ILC method for linear systems with unknown dynamics under the presence of bounded output disturbances and input box-constraints. A method based on the non-parametric data-driven representation of the system dynamics is compared to a subspace identification method. The non-parametric data-driven method shows a promise in dealing with structured disturbances. On the other hand, since the data-driven representation and the online measurement are inexact, we design the ILC based on the inexact classical gradient and Nesterov accelerated gradient methods with proven convergence results. Moreover, we propose a hybrid ILC method which has a faster convergence speed and a better asymptotic performance.

\section*{Acknowledgements}

Ivan Markovsky is an ICREA professor at the International Centre for Numerical Methods in Engineering, Barcelona. The research leading to these results has received funding from: Fond for Scientific Research Vlaanderen (FWO) projects G033822N, G081222N, G0A0920N. European Union’s Horizon 2020 research and innovation programme under the Marie Skłodowska-Curie grant agreement No. 953348.

\bibliographystyle{apacite}
\bibliography{interactapasample}

\begin{thebibliography}{}

\bibitem [\protect \citeauthoryear {%
Barton%
\ \BBA {} Alleyne%
}{%
Barton%
\ \BBA {} Alleyne%
}{%
{\protect \APACyear {2011}}%
}]{%
Barton_2011}
\APACinsertmetastar {%
Barton_2011}%
\begin{APACrefauthors}%
Barton, K\BPBI L.%
\BCBT {}\ \BBA {} Alleyne, A\BPBI G.%
\end{APACrefauthors}%
\unskip\
\newblock
\APACrefYearMonthDay{2011}{}{}.
\newblock
{\BBOQ}\APACrefatitle {A norm optimal approach to time-varying {ILC} with
  application to a multi-axis robotic testbed} {A norm optimal approach to
  time-varying {ILC} with application to a multi-axis robotic testbed}.{\BBCQ}
\newblock
\APACjournalVolNumPages{IEEE Transactions on Control Systems
  Technology}{19}{1}{166-180}.
\PrintBackRefs{\CurrentBib}

\bibitem [\protect \citeauthoryear {%
Bodson%
, Jensen%
\BCBL {}\ \BBA {} Douglas%
}{%
Bodson%
\ \protect \BOthers {.}}{%
{\protect \APACyear {2001}}%
}]{%
Bodson_2001}
\APACinsertmetastar {%
Bodson_2001}%
\begin{APACrefauthors}%
Bodson, M.%
, Jensen, J\BPBI S.%
\BCBL {}\ \BBA {} Douglas, S\BPBI C.%
\end{APACrefauthors}%
\unskip\
\newblock
\APACrefYearMonthDay{2001}{}{}.
\newblock
{\BBOQ}\APACrefatitle {Active noise control for periodic disturbances} {Active
  noise control for periodic disturbances}.{\BBCQ}
\newblock
\APACjournalVolNumPages{IEEE Transactions on Control Systems
  Technology}{9}{1}{200-205}.
\PrintBackRefs{\CurrentBib}

\bibitem [\protect \citeauthoryear {%
Bristow%
\ \BBA {} Alleyne%
}{%
Bristow%
\ \BBA {} Alleyne%
}{%
{\protect \APACyear {2006}}%
}]{%
Bristow_2006_2}
\APACinsertmetastar {%
Bristow_2006_2}%
\begin{APACrefauthors}%
Bristow, D\BPBI A.%
\BCBT {}\ \BBA {} Alleyne, A\BPBI G.%
\end{APACrefauthors}%
\unskip\
\newblock
\APACrefYearMonthDay{2006}{}{}.
\newblock
{\BBOQ}\APACrefatitle {A high precision motion control system with application
  to microscale robotic deposition} {A high precision motion control system
  with application to microscale robotic deposition}.{\BBCQ}
\newblock
\APACjournalVolNumPages{IEEE Transactions on Control Systems
  Technology}{14}{6}{1008-1020}.
\PrintBackRefs{\CurrentBib}

\bibitem [\protect \citeauthoryear {%
Carlet%
, Favato%
, Bolognani%
\BCBL {}\ \BBA {} Dörfler%
}{%
Carlet%
\ \protect \BOthers {.}}{%
{\protect \APACyear {2020}}%
}]{%
Carlet_2020}
\APACinsertmetastar {%
Carlet_2020}%
\begin{APACrefauthors}%
Carlet, P\BPBI G.%
, Favato, A.%
, Bolognani, S.%
\BCBL {}\ \BBA {} Dörfler, F.%
\end{APACrefauthors}%
\unskip\
\newblock
\APACrefYearMonthDay{2020}{}{}.
\newblock
{\BBOQ}\APACrefatitle {Data-driven predictive current control for synchronous
  motor drives} {Data-driven predictive current control for synchronous motor
  drives}.{\BBCQ}
\newblock
\BIn{} \APACrefbtitle {2020 {IEEE} {E}nergy {C}onversion {C}ongress and
  {E}xposition ({ECCE})} {2020 {IEEE} {E}nergy {C}onversion {C}ongress and
  {E}xposition ({ECCE})}\ (\BPG~213-218).
\newblock
\APACaddressPublisher{Detroit, MI, USA}{}.
\PrintBackRefs{\CurrentBib}

\bibitem [\protect \citeauthoryear {%
Chi%
, Hou%
, Huang%
\BCBL {}\ \BBA {} Jin%
}{%
Chi%
\ \protect \BOthers {.}}{%
{\protect \APACyear {2015}}%
}]{%
Chi_2015}
\APACinsertmetastar {%
Chi_2015}%
\begin{APACrefauthors}%
Chi, R.%
, Hou, Z.%
, Huang, B.%
\BCBL {}\ \BBA {} Jin, S.%
\end{APACrefauthors}%
\unskip\
\newblock
\APACrefYearMonthDay{2015}{}{}.
\newblock
{\BBOQ}\APACrefatitle {A unified data-driven design framework of
  optimality-based generalized iterative learning control} {A unified
  data-driven design framework of optimality-based generalized iterative
  learning control}.{\BBCQ}
\newblock
\APACjournalVolNumPages{Computers \& Chemical Engineering}{77}{}{10-23}.
\PrintBackRefs{\CurrentBib}

\bibitem [\protect \citeauthoryear {%
Chi%
, Hou%
, Jin%
\BCBL {}\ \BBA {} Huang%
}{%
Chi%
\ \protect \BOthers {.}}{%
{\protect \APACyear {2018}}%
}]{%
Chi_2018}
\APACinsertmetastar {%
Chi_2018}%
\begin{APACrefauthors}%
Chi, R.%
, Hou, Z.%
, Jin, S.%
\BCBL {}\ \BBA {} Huang, B.%
\end{APACrefauthors}%
\unskip\
\newblock
\APACrefYearMonthDay{2018}{}{}.
\newblock
{\BBOQ}\APACrefatitle {Computationally efficient data-driven higher order
  optimal iterative learning control} {Computationally efficient data-driven
  higher order optimal iterative learning control}.{\BBCQ}
\newblock
\APACjournalVolNumPages{IEEE Transactions on Neural Networks and Learning
  Systems}{29}{}{5971 - 5980}.
\PrintBackRefs{\CurrentBib}

\bibitem [\protect \citeauthoryear {%
Chu%
\ \BBA {} Owens%
}{%
Chu%
\ \BBA {} Owens%
}{%
{\protect \APACyear {2009}}%
}]{%
Chu_2009}
\APACinsertmetastar {%
Chu_2009}%
\begin{APACrefauthors}%
Chu, B.%
\BCBT {}\ \BBA {} Owens, D\BPBI H.%
\end{APACrefauthors}%
\unskip\
\newblock
\APACrefYearMonthDay{2009}{}{}.
\newblock
{\BBOQ}\APACrefatitle {Accelerated norm-optimal iterative learning control
  algorithms using successive projection} {Accelerated norm-optimal iterative
  learning control algorithms using successive projection}.{\BBCQ}
\newblock
\APACjournalVolNumPages{International Journal of Control}{82}{8}{1469–1484}.
\PrintBackRefs{\CurrentBib}

\bibitem [\protect \citeauthoryear {%
Demmel%
}{%
Demmel%
}{%
{\protect \APACyear {1987}}%
}]{%
Demmel_1987}
\APACinsertmetastar {%
Demmel_1987}%
\begin{APACrefauthors}%
Demmel, J\BPBI W.%
\end{APACrefauthors}%
\unskip\
\newblock
\APACrefYearMonthDay{1987}{}{}.
\newblock
{\BBOQ}\APACrefatitle {The Smallest Perturbation of a Submatrix which Lowers
  the Rank and Constrained Total Least Squares Problems} {The smallest
  perturbation of a submatrix which lowers the rank and constrained total least
  squares problems}.{\BBCQ}
\newblock
\APACjournalVolNumPages{SIAM Journal on Numerical Analysis}{24}{1}{199-206}.
\PrintBackRefs{\CurrentBib}

\bibitem [\protect \citeauthoryear {%
de Rozario%
\ \BBA {} Oomen%
}{%
de Rozario%
\ \BBA {} Oomen%
}{%
{\protect \APACyear {2019}}%
}]{%
Rozario_2019}
\APACinsertmetastar {%
Rozario_2019}%
\begin{APACrefauthors}%
de Rozario, R.%
\BCBT {}\ \BBA {} Oomen, T.%
\end{APACrefauthors}%
\unskip\
\newblock
\APACrefYearMonthDay{2019}{}{}.
\newblock
{\BBOQ}\APACrefatitle {Data-driven iterative inversion-based control: Achieving
  robustness through nonlinear learning} {Data-driven iterative inversion-based
  control: Achieving robustness through nonlinear learning}.{\BBCQ}
\newblock
\APACjournalVolNumPages{Automatica}{107}{}{342-352}.
\PrintBackRefs{\CurrentBib}

\bibitem [\protect \citeauthoryear {%
Devolder%
, Glineur%
\BCBL {}\ \BBA {} Nesterov%
}{%
Devolder%
\ \protect \BOthers {.}}{%
{\protect \APACyear {2014}}%
}]{%
Devolder_2014}
\APACinsertmetastar {%
Devolder_2014}%
\begin{APACrefauthors}%
Devolder, O.%
, Glineur, F.%
\BCBL {}\ \BBA {} Nesterov, Y.%
\end{APACrefauthors}%
\unskip\
\newblock
\APACrefYearMonthDay{2014}{}{}.
\newblock
{\BBOQ}\APACrefatitle {First-order methods of smooth convex optimization with
  inexact oracle} {First-order methods of smooth convex optimization with
  inexact oracle}.{\BBCQ}
\newblock
\APACjournalVolNumPages{Mathematical Programming}{146}{}{37–75}.
\PrintBackRefs{\CurrentBib}

\bibitem [\protect \citeauthoryear {%
Elokda%
, Coulson%
, Beuchat%
, Lygeros%
\BCBL {}\ \BBA {} Dörfler%
}{%
Elokda%
\ \protect \BOthers {.}}{%
{\protect \APACyear {2021}}%
}]{%
Elokda_2021}
\APACinsertmetastar {%
Elokda_2021}%
\begin{APACrefauthors}%
Elokda, E.%
, Coulson, J.%
, Beuchat, P\BPBI N.%
, Lygeros, J.%
\BCBL {}\ \BBA {} Dörfler, F.%
\end{APACrefauthors}%
\unskip\
\newblock
\APACrefYearMonthDay{2021}{}{}.
\newblock
{\BBOQ}\APACrefatitle {Data‐enabled predictive control for quadcopters}
  {Data‐enabled predictive control for quadcopters}.{\BBCQ}
\newblock
\APACjournalVolNumPages{International Journal of Robust and Nonlinear
  Control}{31}{}{8916–8936}.
\PrintBackRefs{\CurrentBib}

\bibitem [\protect \citeauthoryear {%
Geng%
\ \BBA {} Ruan%
}{%
Geng%
\ \BBA {} Ruan%
}{%
{\protect \APACyear {2015}}%
}]{%
Geng_2015}
\APACinsertmetastar {%
Geng_2015}%
\begin{APACrefauthors}%
Geng, Y.%
\BCBT {}\ \BBA {} Ruan, X.%
\end{APACrefauthors}%
\unskip\
\newblock
\APACrefYearMonthDay{2015}{}{}.
\newblock
{\BBOQ}\APACrefatitle {Quasi-{N}ewton-type optimized iterative learning control
  for discrete linear time invariant systems} {Quasi-{N}ewton-type optimized
  iterative learning control for discrete linear time invariant
  systems}.{\BBCQ}
\newblock
\APACjournalVolNumPages{Control Theory and Technology}{13}{}{256–265}.
\PrintBackRefs{\CurrentBib}

\bibitem [\protect \citeauthoryear {%
Gu%
, Tian%
\BCBL {}\ \BBA {} Chen%
}{%
Gu%
\ \protect \BOthers {.}}{%
{\protect \APACyear {2019}}%
}]{%
Gu_2019}
\APACinsertmetastar {%
Gu_2019}%
\begin{APACrefauthors}%
Gu, P.%
, Tian, S.%
\BCBL {}\ \BBA {} Chen, Y\BPBI Q.%
\end{APACrefauthors}%
\unskip\
\newblock
\APACrefYearMonthDay{2019}{}{}.
\newblock
{\BBOQ}\APACrefatitle {Iterative Learning Control Based on {N}esterov
  Accelerated Gradient Method} {Iterative learning control based on {N}esterov
  accelerated gradient method}.{\BBCQ}
\newblock
\APACjournalVolNumPages{IEEE Access}{7}{}{115836-115842}.
\PrintBackRefs{\CurrentBib}

\bibitem [\protect \citeauthoryear {%
Harte%
, Hätönen%
\BCBL {}\ \BBA {} Owens%
}{%
Harte%
\ \protect \BOthers {.}}{%
{\protect \APACyear {2005}}%
}]{%
Harte_2005}
\APACinsertmetastar {%
Harte_2005}%
\begin{APACrefauthors}%
Harte, T\BPBI J.%
, Hätönen, J\BPBI J.%
\BCBL {}\ \BBA {} Owens, D\BPBI H.%
\end{APACrefauthors}%
\unskip\
\newblock
\APACrefYearMonthDay{2005}{}{}.
\newblock
{\BBOQ}\APACrefatitle {Discrete-time inverse model-based iterative learning
  control: stability, monotonicity and robustness} {Discrete-time inverse
  model-based iterative learning control: stability, monotonicity and
  robustness}.{\BBCQ}
\newblock
\APACjournalVolNumPages{International Journal of Control}{78}{8}{577–586}.
\PrintBackRefs{\CurrentBib}

\bibitem [\protect \citeauthoryear {%
Janczak%
, Rogers%
\BCBL {}\ \BBA {} Cai%
}{%
Janczak%
\ \protect \BOthers {.}}{%
{\protect \APACyear {2013}}%
}]{%
Janczak_2013}
\APACinsertmetastar {%
Janczak_2013}%
\begin{APACrefauthors}%
Janczak, A.%
, Rogers, E.%
\BCBL {}\ \BBA {} Cai, Z.%
\end{APACrefauthors}%
\unskip\
\newblock
\APACrefYearMonthDay{2013}{}{}.
\newblock
{\BBOQ}\APACrefatitle {Subspace identification of process dynamics for
  iterative learning control} {Subspace identification of process dynamics for
  iterative learning control}.{\BBCQ}
\newblock
\BIn{} \APACrefbtitle {Proceedings of the 8th {I}nternational {W}orkshop on
  {M}ultidimensional {S}ystems} {Proceedings of the 8th {I}nternational
  {W}orkshop on {M}ultidimensional {S}ystems}\ (\BPG~39-44).
\newblock
\APACaddressPublisher{Erlangen, Germany}{}.
\PrintBackRefs{\CurrentBib}

\bibitem [\protect \citeauthoryear {%
Jiang%
, Chen%
\BCBL {}\ \BBA {} Chu%
}{%
Jiang%
\ \protect \BOthers {.}}{%
{\protect \APACyear {2023}}%
}]{%
Jiang_2023}
\APACinsertmetastar {%
Jiang_2023}%
\begin{APACrefauthors}%
Jiang, Z.%
, Chen, B.%
\BCBL {}\ \BBA {} Chu, B.%
\end{APACrefauthors}%
\unskip\
\newblock
\APACrefYearMonthDay{2023}{}{}.
\newblock
{\BBOQ}\APACrefatitle {Data-driven norm optimal iterative learning control for
  point-to-point tasks} {Data-driven norm optimal iterative learning control
  for point-to-point tasks}.{\BBCQ}
\newblock
\APACjournalVolNumPages{IFAC-PapersOnLine}{56}{2}{1051–1056}.
\PrintBackRefs{\CurrentBib}

\bibitem [\protect \citeauthoryear {%
Jiang%
\ \BBA {} Chu%
}{%
Jiang%
\ \BBA {} Chu%
}{%
{\protect \APACyear {2022}}%
}]{%
Jiang_2022}
\APACinsertmetastar {%
Jiang_2022}%
\begin{APACrefauthors}%
Jiang, Z.%
\BCBT {}\ \BBA {} Chu, B.%
\end{APACrefauthors}%
\unskip\
\newblock
\APACrefYearMonthDay{2022}{}{}.
\newblock
{\BBOQ}\APACrefatitle {Norm Optimal Iterative Learning Control: {A} Data-Driven
  Approach} {Norm optimal iterative learning control: {A} data-driven
  approach}.{\BBCQ}
\newblock
\APACjournalVolNumPages{IFAC-PapersOnLine}{55}{12}{482–487}.
\PrintBackRefs{\CurrentBib}

\bibitem [\protect \citeauthoryear {%
Markovsky%
\ \BBA {} Dörfler%
}{%
Markovsky%
\ \BBA {} Dörfler%
}{%
{\protect \APACyear {2021}}%
}]{%
Markovsky_2021}
\APACinsertmetastar {%
Markovsky_2021}%
\begin{APACrefauthors}%
Markovsky, I.%
\BCBT {}\ \BBA {} Dörfler, F.%
\end{APACrefauthors}%
\unskip\
\newblock
\APACrefYearMonthDay{2021}{}{}.
\newblock
{\BBOQ}\APACrefatitle {Behavioral systems theory in data-driven analysis,
  signal processing, and control} {Behavioral systems theory in data-driven
  analysis, signal processing, and control}.{\BBCQ}
\newblock
\APACjournalVolNumPages{Annual Reviews in Control}{52}{}{42–64}.
\PrintBackRefs{\CurrentBib}

\bibitem [\protect \citeauthoryear {%
Markovsky%
\ \BBA {} Dörfler%
}{%
Markovsky%
\ \BBA {} Dörfler%
}{%
{\protect \APACyear {2023}}%
}]{%
Markovsky_2023}
\APACinsertmetastar {%
Markovsky_2023}%
\begin{APACrefauthors}%
Markovsky, I.%
\BCBT {}\ \BBA {} Dörfler, F.%
\end{APACrefauthors}%
\unskip\
\newblock
\APACrefYearMonthDay{2023}{}{}.
\newblock
{\BBOQ}\APACrefatitle {Identifiability in the behavioral setting}
  {Identifiability in the behavioral setting}.{\BBCQ}
\newblock
\APACjournalVolNumPages{IEEE Transactions on Automatic
  Control}{68}{3}{1667-1677}.
\PrintBackRefs{\CurrentBib}

\bibitem [\protect \citeauthoryear {%
Markovsky%
, Liu%
\BCBL {}\ \BBA {} Takeda%
}{%
Markovsky%
\ \protect \BOthers {.}}{%
{\protect \APACyear {2020}}%
}]{%
Markovsky_2020}
\APACinsertmetastar {%
Markovsky_2020}%
\begin{APACrefauthors}%
Markovsky, I.%
, Liu, T.%
\BCBL {}\ \BBA {} Takeda, A.%
\end{APACrefauthors}%
\unskip\
\newblock
\APACrefYearMonthDay{2020}{}{}.
\newblock
{\BBOQ}\APACrefatitle {Data-driven structured noise filtering via common
  dynamics estimation} {Data-driven structured noise filtering via common
  dynamics estimation}.{\BBCQ}
\newblock
\APACjournalVolNumPages{IEEE Transactions on Signal
  Processing}{68}{}{3064-3073}.
\PrintBackRefs{\CurrentBib}

\bibitem [\protect \citeauthoryear {%
Markovsky%
\ \BBA {} Rapisarda%
}{%
Markovsky%
\ \BBA {} Rapisarda%
}{%
{\protect \APACyear {2008}}%
}]{%
Markovsky_2008}
\APACinsertmetastar {%
Markovsky_2008}%
\begin{APACrefauthors}%
Markovsky, I.%
\BCBT {}\ \BBA {} Rapisarda, P.%
\end{APACrefauthors}%
\unskip\
\newblock
\APACrefYearMonthDay{2008}{}{}.
\newblock
{\BBOQ}\APACrefatitle {Data-driven simulation and control} {Data-driven
  simulation and control}.{\BBCQ}
\newblock
\APACjournalVolNumPages{International Journal of Control}{81}{12}{1946–1959}.
\PrintBackRefs{\CurrentBib}

\bibitem [\protect \citeauthoryear {%
Nesterov%
}{%
Nesterov%
}{%
{\protect \APACyear {1983}}%
}]{%
Nesterov_1983}
\APACinsertmetastar {%
Nesterov_1983}%
\begin{APACrefauthors}%
Nesterov, Y.%
\end{APACrefauthors}%
\unskip\
\newblock
\APACrefYearMonthDay{1983}{}{}.
\newblock
{\BBOQ}\APACrefatitle {A method of solving a convex programming problem with
  convergence rate {O}$(1/k^{2})$} {A method of solving a convex programming
  problem with convergence rate {O}$(1/k^{2})$}.{\BBCQ}
\newblock
\APACjournalVolNumPages{Doklady Akademii Nauk}{27}{2}{372–376}.
\PrintBackRefs{\CurrentBib}

\bibitem [\protect \citeauthoryear {%
Nesterov%
}{%
Nesterov%
}{%
{\protect \APACyear {2005}}%
}]{%
Nesterov_2005}
\APACinsertmetastar {%
Nesterov_2005}%
\begin{APACrefauthors}%
Nesterov, Y.%
\end{APACrefauthors}%
\unskip\
\newblock
\APACrefYearMonthDay{2005}{}{}.
\newblock
{\BBOQ}\APACrefatitle {Smooth minimization of non-smooth functions} {Smooth
  minimization of non-smooth functions}.{\BBCQ}
\newblock
\APACjournalVolNumPages{Mathematical Programming}{103}{}{127–152}.
\PrintBackRefs{\CurrentBib}

\bibitem [\protect \citeauthoryear {%
Nijsse%
, Verhaegen%
\BCBL {}\ \BBA {} Doelman%
}{%
Nijsse%
\ \protect \BOthers {.}}{%
{\protect \APACyear {2001}}%
}]{%
Nijsse_2001}
\APACinsertmetastar {%
Nijsse_2001}%
\begin{APACrefauthors}%
Nijsse, G.%
, Verhaegen, M.%
\BCBL {}\ \BBA {} Doelman, N\BPBI J.%
\end{APACrefauthors}%
\unskip\
\newblock
\APACrefYearMonthDay{2001}{}{}.
\newblock
{\BBOQ}\APACrefatitle {A new subspace based approach to iterative learning
  control} {A new subspace based approach to iterative learning
  control}.{\BBCQ}
\newblock
\BIn{} \APACrefbtitle {2001 {E}uropean {C}ontrol {C}onference ({ECC})} {2001
  {E}uropean {C}ontrol {C}onference ({ECC})}\ (\BPG~3375-3380).
\newblock
\APACaddressPublisher{Porto, Portugal}{}.
\PrintBackRefs{\CurrentBib}

\bibitem [\protect \citeauthoryear {%
Park%
}{%
Park%
}{%
{\protect \APACyear {1999}}%
}]{%
Park_1999}
\APACinsertmetastar {%
Park_1999}%
\begin{APACrefauthors}%
Park, K\BPBI H.%
\end{APACrefauthors}%
\unskip\
\newblock
\APACrefYearMonthDay{1999}{}{}.
\newblock
{\BBOQ}\APACrefatitle {A study on the robustness of a {PID}-type iterative
  learning controller against initial state error} {A study on the robustness
  of a {PID}-type iterative learning controller against initial state
  error}.{\BBCQ}
\newblock
\APACjournalVolNumPages{International Journal of Systems
  Science}{30}{1}{49-59}.
\PrintBackRefs{\CurrentBib}

\bibitem [\protect \citeauthoryear {%
Ratcliffe%
, Hätönen%
, Lewin%
, Rogers%
\BCBL {}\ \BBA {} Owens%
}{%
Ratcliffe%
\ \protect \BOthers {.}}{%
{\protect \APACyear {2008}}%
}]{%
Ratcliffe_2008}
\APACinsertmetastar {%
Ratcliffe_2008}%
\begin{APACrefauthors}%
Ratcliffe, J\BPBI D.%
, Hätönen, J\BPBI J.%
, Lewin, P\BPBI L.%
, Rogers, E.%
\BCBL {}\ \BBA {} Owens, D\BPBI H.%
\end{APACrefauthors}%
\unskip\
\newblock
\APACrefYearMonthDay{2008}{}{}.
\newblock
{\BBOQ}\APACrefatitle {Robustness analysis of an adjoint optimal iterative
  learning controller with experimental verification} {Robustness analysis of
  an adjoint optimal iterative learning controller with experimental
  verification}.{\BBCQ}
\newblock
\APACjournalVolNumPages{International Journal of Robust and Nonlinear
  Control}{18}{}{1089–1113}.
\PrintBackRefs{\CurrentBib}

\bibitem [\protect \citeauthoryear {%
Tayebi%
}{%
Tayebi%
}{%
{\protect \APACyear {2004}}%
}]{%
Tayebi_2004}
\APACinsertmetastar {%
Tayebi_2004}%
\begin{APACrefauthors}%
Tayebi, A.%
\end{APACrefauthors}%
\unskip\
\newblock
\APACrefYearMonthDay{2004}{}{}.
\newblock
{\BBOQ}\APACrefatitle {Adaptive iterative learning control for robot
  manipulators} {Adaptive iterative learning control for robot
  manipulators}.{\BBCQ}
\newblock
\APACjournalVolNumPages{Automatica}{7}{}{1195-1203}.
\PrintBackRefs{\CurrentBib}

\bibitem [\protect \citeauthoryear {%
Tien%
, Zou%
\BCBL {}\ \BBA {} Devasia%
}{%
Tien%
\ \protect \BOthers {.}}{%
{\protect \APACyear {2005}}%
}]{%
Tien_2005}
\APACinsertmetastar {%
Tien_2005}%
\begin{APACrefauthors}%
Tien, S.%
, Zou, Q.%
\BCBL {}\ \BBA {} Devasia, S.%
\end{APACrefauthors}%
\unskip\
\newblock
\APACrefYearMonthDay{2005}{}{}.
\newblock
{\BBOQ}\APACrefatitle {Iterative control of dynamics-coupling-caused errors in
  piezoscanners during high-speed {AFM} operation} {Iterative control of
  dynamics-coupling-caused errors in piezoscanners during high-speed {AFM}
  operation}.{\BBCQ}
\newblock
\APACjournalVolNumPages{IEEE Transactions on Control Systems
  Technology}{13}{6}{921-931}.
\PrintBackRefs{\CurrentBib}

\bibitem [\protect \citeauthoryear {%
van Overschee%
\ \BBA {} de Moor%
}{%
van Overschee%
\ \BBA {} de Moor%
}{%
{\protect \APACyear {2012}}%
}]{%
Overschee_2012}
\APACinsertmetastar {%
Overschee_2012}%
\begin{APACrefauthors}%
van Overschee, P.%
\BCBT {}\ \BBA {} de Moor, B.%
\end{APACrefauthors}%
\unskip\
\newblock
\APACrefYear{2012}.
\newblock
\APACrefbtitle {Subspace identification for linear systems: {T}heory,
  {I}mplementation, {A}pplications} {Subspace identification for linear
  systems: {T}heory, {I}mplementation, {A}pplications}.
\newblock
\APACaddressPublisher{The Netherlanders}{Kluwer Academic Publishers Group}.
\PrintBackRefs{\CurrentBib}

\bibitem [\protect \citeauthoryear {%
Wang%
\ \BBA {} Meng%
}{%
Wang%
\ \BBA {} Meng%
}{%
{\protect \APACyear {2023}}%
}]{%
Wang_2023}
\APACinsertmetastar {%
Wang_2023}%
\begin{APACrefauthors}%
Wang, C.%
\BCBT {}\ \BBA {} Meng, D.%
\end{APACrefauthors}%
\unskip\
\newblock
\APACrefYearMonthDay{2023}{}{}.
\newblock
{\BBOQ}\APACrefatitle {Data-Based Optimization Control for Learning Systems}
  {Data-based optimization control for learning systems}.{\BBCQ}
\newblock
\APACjournalVolNumPages{IEEE Transactions on Circuits and Systems II: Express
  Briefs}{70}{7}{2560–2564}.
\PrintBackRefs{\CurrentBib}

\bibitem [\protect \citeauthoryear {%
Willems%
, Rapisarda%
, Markovsky%
\BCBL {}\ \BBA {} de Moor%
}{%
Willems%
\ \protect \BOthers {.}}{%
{\protect \APACyear {2005}}%
}]{%
Willems_2005}
\APACinsertmetastar {%
Willems_2005}%
\begin{APACrefauthors}%
Willems, J\BPBI C.%
, Rapisarda, P.%
, Markovsky, I.%
\BCBL {}\ \BBA {} de Moor, B.%
\end{APACrefauthors}%
\unskip\
\newblock
\APACrefYearMonthDay{2005}{}{}.
\newblock
{\BBOQ}\APACrefatitle {A note on persistency of excitation} {A note on
  persistency of excitation}.{\BBCQ}
\newblock
\APACjournalVolNumPages{Systems \& Control Letters}{54}{4}{325–329}.
\PrintBackRefs{\CurrentBib}

\end{thebibliography}

\end{document}